\documentclass[11pt,reqno]{amsart}

\usepackage{etex}
\usepackage[utf8]{inputenc}
\usepackage{amsfonts}
\usepackage{amsmath}
\usepackage{amssymb}
\usepackage{amsthm}
\usepackage{mathrsfs}
\usepackage{stmaryrd}
\usepackage{color}
\usepackage[english]{babel}
\usepackage[T1]{fontenc}
\usepackage{url}
\usepackage{graphicx}
\usepackage[pdftex,colorlinks,backref=page,citecolor=blue]{hyperref}
\usepackage{caption}
\usepackage{enumerate}
\usepackage{epstopdf}
\usepackage{hyphenat}
\usepackage{float}
\usepackage{csquotes}
\usepackage{mathtools}
\usepackage{indentfirst}
\usepackage[export]{adjustbox}
\usepackage{tikz}
\usepackage{color}
\usepackage[all]{xy}
\usetikzlibrary{matrix,arrows,decorations.pathmorphing}
\usepackage{booktabs}
\usepackage{array}
\usepackage{bm}
\usepackage{tcolorbox}
\usepackage{booktabs}
\usepackage{longtable}
\usepackage{tikz-cd}

\newcolumntype{P}[1]{>{\centering\arraybackslash}p{#1}}

\allowdisplaybreaks

\newtheorem{theorem}{Theorem}[section]

\newtheorem{lemma}[theorem]{Lemma}
\newtheorem{proposition}[theorem]{Proposition}
\newtheorem{corollary}[theorem]{Corollary}

\newcommand{\Lex}{\operatorname{Lex}}

\newcommand{\DegRevLex}{\operatorname{DegRevLex}}

\newcommand{\LT}{\operatorname{LT}}
\newcommand{\T}{\operatorname{T}}

\newcommand{\Supp}{\operatorname{Supp}}
\newcommand{\start}{\operatorname{start}}
\newcommand{\finish}{\operatorname{finish}}

\newcommand{\Min}{\operatorname{Min}}

\newcommand{\SolvingDegree}{\operatorname{sd}}

\newcommand{\LexGB}{\operatorname{LexGB}}
\newcommand{\sat}{\operatorname{sat}}

\theoremstyle{definition}

\newtheorem{problem}[theorem]{Problem}
\newtheorem{remark}[theorem]{Remark}

\newtheorem{example}[theorem]{Example}
\newtheorem{definition}[theorem]{Definition}

\newtheorem*{algorithmm*}{Polynomial System}

\makeatletter
\def\l@subsection{\@tocline{2}{0pt}{2.5pc}{5pc}{}}
\makeatother

\author{Alessio Caminata}
\address{Dipartimento di Matematica, Universit\`a di Genova\\ via Dodecaneso 35, 16146, Genova, Italy}
\email{alessio.caminata@unige.it}

\author{Andrea Sanguineti}
\address{Dipartimento di Matematica, Universit\`a di Genova\\ via Dodecaneso 35, 16146, Genova, Italy}
\email{andrea.o.sanguineti@gmail.com}

\author{Silvia Sconza}
\address{Institute of Mathematics, University of Zurich, Zurich, Switzerland }
\email{silvia.sconza@math.uzh.ch}

\title[Algebraic Modelings of the Supersingular Isogeny Problem]{Algebraic Modelings of the Supersingular Isogeny Problem}

\keywords{Supersingular Isogeny Problem, Gr\"obner bases, solving degree, algebraic modeling.}

\begin{document}

\begin{abstract}
We present a new algebraic modeling of the Supersingular Isogeny Problem as a system of multivariate polynomial equations, in the case where the elliptic curves are connected by an isogeny whose degree is a power of $2$ or $3$. This modeling relies on Renes  formulas for elliptic curves in Montgomery form (degree $2$) or triangular form (degree $3$). We investigate several algebraic properties of these systems: we prove that they are zero-dimensional, compute the dimension of their highest degree part, and show that they are not in generic coordinates. Experimental results show that solving these systems via Gr\"obner basis techniques is significantly faster than solving the algebraic modeling with modular polynomials.
\end{abstract}

\maketitle

\section*{Introduction}
We fix a prime number $p>3$ and let $E$ and $E'$ be two supersingular elliptic curves defined over $\mathbb{F}_{p^2}$ such that there exists a degree $d$ isogeny $\varphi:E\rightarrow E'$. The \emph{Supersingular Isogeny Problem (SIP)} asks to find $\varphi$. When $d$ and $p$ are sufficiently large, a random instance of this problem is believed to be computationally difficult to solve, even using a quantum computer. For this reason, it has been used as the underlying problem for several post-quantum cryptographic schemes, such as the signature scheme SQIsign \cite{SQIsign}, which, at the time of writing, is admitted to Round 3 of the NIST call for post-quantum signature schemes. For this reason, the SIP has been widely studied and several algorithms and methods have been proposed in the literature.  We note that, since the degree of isogenies is multiplicative under composition, the main difficulty in solving the SIP arises when the degree $d$ is a power of a prime $\ell\neq p$. In this setting, the most efficient approaches exploit structural properties of the graph of $\ell$-isogenies between supersingular elliptic curves \cite{corte2022accelerating, CS20, DG16, GPS20}.

Given the growing importance of post-quantum cryptography, both the cryptanalysis of proposed schemes and the study of the hardness of their underlying problems have received increasing attention in recent years. A comprehensive security assessment requires considering attacks originating from different areas of post-quantum cryptography. In this context, algebraic modeling and algebraic attacks have emerged as important tools. In the literature, several algebraic models have been proposed and studied for problems arising in code-based cryptography \cite{CMMPR25, DGRR26, MPS23, Saeed} and lattice-based cryptography \cite{ACFRP15, AroraGe, Steiner24}. In contrast, algebraic modeling of problems in isogeny-based cryptography has received comparatively less attention \cite{Taketal}.

A natural approach to modeling the Supersingular Isogeny Problem as a system of multivariate polynomial equations is via modular polynomials \cite[\S10]{Was08}. Given a positive integer $N\neq p$ the $N$-th modular polynomial $\Phi_{N}(X,Y)$ is a polynomial (with integer coefficients) with the property that  $\Phi_N(j_1,j_2) = 0$ if and only if $j_1,j_2$ are the $j$-invariants of elliptic curves that are related by an isogeny of degree $N$. So, if $N=\ell^e$, where $\ell$ is a prime number, given the $j$-invariants $j_{\start},j_{\finish}$ of two elliptic curves $E_{\start},E_{\finish}$ connected by a $N$-isogeny $\varphi$, we can reduce the SIP to several smaller instances of the SIP in degree $\ell$ by solving a multivariate polynomial system (see Polynomial system~\hyperlink{system1}{1}) with modular polynomials, whose solutions represent all possible paths from $E_{\start}$ to $E_{\finish}$ in the $\ell$-isogeny graph. This approach has been studied in the paper \cite{Taketal}. There, the authors propose the algebraic model and study the complexity of solving the corresponding polynomial system via Gr\"obner basis techniques.

In this paper, we propose two distinct algebraic models for the SIP: one for isogenies whose degree is a power of $2$, and one for isogenies whose degree is a power of $3$. These two cases are the most relevant for cryptographic applications. The main ideas, which originate from the works of Renes \cite{Renes}, Burdges--De Feo \cite{Burdges-DeFeo}, and Costello--Hisil \cite{Costello-Hisil}, are as follows.  
If an elliptic curve E is given in Montgomery form, then the $2$-isogenies with domain E can be described explicitly in terms of the $x$-coordinates of the $2$-torsion points of E. More precisely, if $\alpha\in\mathbb{F}_{p^2}$ is the $x$-coordinate of a $2$-torsion point of $E$, then one can construct a $2$-isogeny $\varphi_{\alpha}:E\rightarrow E({\alpha})$, where $E({\alpha})$ can also be expressed in Montgomery form, allowing us to iteratively describe also the $2$-isogenies with domain $E(\alpha)$. Moreover, the $2$-isogeny with kernel $\langle(0,0)\rangle$ is the dual isogeny of $\varphi_{\alpha}$. By proceeding in this way, we can explicitly construct a non-backtracking path (a sequence in which no two consecutive isogenies are dual to each other) of $2$-isogenies starting from $E$.
We use these formulas to derive multivariate polynomials, which we call \emph{Renes polynomials} (Definition~\ref{def:Renespolynomials2}), playing a role analogous to that of modular polynomials for curves in Montgomery form. Using these polynomials, we construct a system of polynomial equations modeling the SIP for curves in Montgomery form and isogenies of degree a power of $2$ (Polynomial System~\hyperlink{system2}{2}).
One key difference with respect to the modular polynomials setting is that, once a $j$-invariant is fixed, there exist several non-isomorphic elliptic curves in Montgomery with that $j$-invariant. As a consequence, the resulting polynomial system contains a polynomial of degree $12$, which significantly impacts the efficiency of solving the system. However, this issue can be mitigated by decomposing the problem into six quadratic systems (Polynomial System~\hyperlink{system3}{3}), whose computations can be parallelized.

The case of degree $3$ is analogous, with the key difference that, in order to explicitly describe the $3$-isogenies with domain an elliptic curve $E$, the curve must be expressed in triangular form. As in the degree $2$ setting, the central idea is to ensure that the point $(0,0)$ is a $3$-torsion point. In this framework, the $3$-isogenies with kernel $\langle(0,0)\rangle$ correspond to the dual isogenies along the path, and we can exclude them to obtain a non-backtracking path.
As before, there exist multiple non-isomorphic curves in triangular form sharing the same $j$-invariant. Consequently, we construct $12$ quadratic systems (Polynomial System~\hyperlink{system7}{5}), whose computations can be parallelized.

Once the algebraic models have been constructed, we analyze the complexity of solving them using Gr\"obner basis techniques and compare it with that of systems based on modular polynomials. A key parameter for estimating the complexity of solving a multivariate polynomial system $\mathcal{F}$ via linear-algebra-based Gr\"obner basis algorithms is the solving degree, denoted $\SolvingDegree(\mathcal{F})$. Informally, the solving degree is the highest degree reached by the polynomials appearing during the computation of a Gr\"obner basis using such algorithms.
Estimating this degree without explicitly solving the system, and thus understanding how it grows with respect to the input parameters (such as the number of variables, the number of equations, and their degrees), is notoriously difficult. Several invariants have been introduced to address this problem, including the degree of regularity \cite{BFS04, Salizzoni} and the Castelnuovo-Mumford regularity \cite{Caminata-Gorla2}.
The degree of regularity is defined when the highest-degree homogeneous component $\mathcal{F}^{\mathrm{top}}$ of the system is zero-dimensional. However, for the systems arising from the algebraic modeling of the SIP, both with modular and Renes polynomials, we show that this condition does not hold (Proposition~\ref{prop:Krull-dim}). On the other hand, when a system is in generic coordinates, the solving degree is bounded above by the Castelnuovo-Mumford regularity \cite{Caminata-Gorla2}, a fundamental invariant in commutative algebra and algebraic geometry. Nevertheless, we prove that the systems under consideration are not in generic coordinates (Proposition~\ref{prop:genericcoordinates}).

Finally, we perform computational experiments to compare both the running times and the solving degrees of the polynomial systems arising from modular and Renes polynomials (see Section~\ref{section:experiments}). The results show that the approach based on Renes polynomials is significantly faster and typically yields a strictly smaller solving degree. For instance, for a $24$-bit prime $p$, we are able to solve the SIP for isogenies of degree up to $2^{15}$ with a solving degree of $6$. In contrast, for the same parameters using modular polynomials, computations become infeasible beyond degree $2^{12}$, with a solving degree of~$9$.

\subsection*{Structure of the paper}
In Section~\ref{section:preliminaries}, we recall some definitions, notations, and results on elliptic curves and polynomial system solving via Gr\"obner bases. In Section~\ref{section:modularpolynomials}, we recall the algebraic modeling with modular polynomials from \cite{Taketal}. In Section~\ref{section:Renes2} and Section~\ref{section:Renes3} we present our algebraic modelings with Renes polynomials in degrees $2$ and $3$ respectively. In Section~\ref{section:algebraicanalysis}, we prove some algebraic properties of these systems and in Section~\ref{section:experiments} we perform some experiments to compare them.

\subsection*{Acknowledgments}
We thank Luca De Feo for bringing Renes’ work \cite{Renes} to our attention. We also thank Giulio Codogni, Ignacio M. Jiménez, Guido Lido, and Marzio Mula for several helpful comments on an earlier version of this paper.

\par A. Caminata and A. Sanguineti are supported by the PRIN PNRR 2022 grant P2022J4HRR \enquote{Mathematical Primitives for Post Quantum Digital Signatures} and by the MUR Excellence Department Project awarded to Dipartimento di Matematica, Università di Genova, CUP D33C23001110001.
A. Caminata is supported by the PRIN 2022 grant 2022K48YYP \enquote{Unirationality, Hilbert schemes, and singularities}.

S. Sconza’s research is supported by armasuisse Science and Technology.

\section{Preliminaries}\label{section:preliminaries}

Throughout the paper $p$ denotes a prime number greater than $3$.

\subsection{Isogeny Problem}
An \emph{elliptic curve} $E$ over a finite field $\mathbb{F}_q$ of characteristic $p$ is a projective plane curve which is defined (in affine coordinates) by a \emph{Weierstrass equation} $y^2=x^3+ax+b$ with $a,b\in\mathbb{F}_q$ and \emph{discriminant} $\Delta(E)=4a^3+27b^2\neq0$. Its \emph{$j$-invariant} is defined as $j(E)=1728\frac{4a^3}{4a^3+27b^2}$.
Two elliptic curves have the same $j$-invariant if and only if they are isomorphic (as algebraic curves) over the algebraic closure $\overline{\mathbb{F}}_q$ of $\mathbb{F}_q$.
The set of points of $E$ is endowed with an addition law that makes it into an abelian group whose identity element is $O_E=[0:1:0]$, the unique point at infinity of the curve.
The set $E(\mathbb{F}_q)$ of $\mathbb{F}_q$-rational points of $E$ forms a subgroup.

An \emph{isogeny} $\varphi$ between two elliptic curves $E$ and $E'$ is a morphism $\varphi\colon E\rightarrow E'$, which is also a group homomorphism of the underlying group structures. For example, the \emph{multiplication-by-$m$} map $[m]\colon E \rightarrow E$ is an isogeny. Its kernel, denoted by $E[m]$, is the set of $m$-torsion points of $E$. The degree of an isogeny is its degree as morphism.  We will write $\ell$-isogeny to indicate an isogeny of degree $\ell$. The degree is multiplicative under composition of isogenies: if $\varphi$ and $\psi$ are isogenies such that $\varphi\circ\psi$ is defined, then $\deg(\varphi\circ\psi)=\deg(\varphi)\cdot\deg(\psi)$.
We say that an isogeny is \emph{cyclic} if its kernel is a cyclic group. 
We recall that an isogeny is uniquely determined by its kernel, up to isomorphism. For example, when the degree of the isogeny is not divisible by the characteristic of the field, we have  $\deg(\varphi)=|\ker(\varphi)|$. In particular, if an isogeny has prime degree, then it is necessarily cyclic.
Finally, we recall that given an isogeny $\varphi\colon E\rightarrow E'$ of degree $\ell$, there exists an isogeny $\widehat{\varphi}\colon E'\rightarrow E$ of degree $\ell$ such that $\varphi\circ\widehat{\varphi}=[\ell]$ and $\widehat{\varphi}\circ\varphi=[\ell]$. Such an isogeny $\widehat{\varphi}$ is called the \emph{dual isogeny} of $\varphi$.

Elliptic curves can be partitioned into two disjoint families: ordinary and supersingular elliptic curves. These can be defined in many equivalent ways. For example, we say that an elliptic curve $E$ is \emph{supersingular} if the subgroup of $p$-torsion points (where $p$ is the characteristic of the base field $\mathbb{F}_q$) is trivial, i.e., $E[p]=\{O_E\}$. Otherwise, $E$ is said to be \emph{ordinary} and in this case $E[p]\cong\mathbb{Z}/p\mathbb{Z}$. A distinctive property of supersingular elliptic curves is that their $j$-invariants are always defined over $\mathbb{F}_{p^2}$. Whether an elliptic curve is ordinary or supersingular is an isogeny invariant. Namely, if $\varphi\colon E\rightarrow E'$ is an isogeny, then $E$ is supersingular (resp. ordinary) if and only if $E'$ is supersingular (resp. ordinary). 

We are interested in the following problem.

\begin{problem}[SIP: Supersingular Isogeny Problem]
	Given $E$ and $E'$ two supersingular elliptic curves defined over $\mathbb{F}_q$, find (if it exists) an isogeny $\varphi\colon E\rightarrow E'$ of degree $\ell$.
\end{problem}

A key observation in the study of this problem is that any isogeny of degree $\ell$ can be decomposed into a sequence of isogenies of prime degree, by factoring $\ell$ into its prime components.  Thus, it is often enough to focus on the Supersingular Isogeny Problem when $\ell$ is a prime power. In practice, chains of isogenies of prime degree $2$ or $3$ are often considered. Notice that in this setup all the elliptic curves in the intermediate steps of the chain are supersingular.

\subsection{Polynomial system solving via Gr\"obner bases}
In this section, we give a brief overview on how to solve polynomial systems by using Gr\"obner bases and the related relevant invariants. 

For consistency with rest of the paper, we fix a finite field $\mathbb{F}_q$ (although everything in this section can be defined over any field), and we work over the polynomial ring  $\mathbb{F}_q[x_1,\dots,x_n]$.
Let $T$ be the set of terms (i.e., monic monomials) of $\mathbb{F}_q[x_1,\dots,x_n]$. A \emph{term order} $<$ on $T$ is a total order which is compatible with the multiplicative structure, that is $m_1 < m_2$ implies $m_1m < m_2m$, and $m_1 \mid m_2$ implies $m_1 \leq m_2$. Given a polynomial $f$ we denote by $\LT_<(f)$ its leading term, i.e., the largest term of $f$ with respect to the chosen term order. Similarly, given an ideal $I$, we denote by $\LT_<(I)=(\LT_<(f): \ f\in I)$, the \emph{initial ideal} of $I$, which is generated by all leading terms of polynomials of $I$. If $I$ is generated by polynomials $f_1,\dots,f_m$, then it holds $(\LT_<(f_1),\dots,\LT_<(f_m))\subseteq\LT_<(I)$, but the inclusion might be strict in general. When equality holds, we say that $f_1,\dots,f_m$ is a \emph{Gr\"obner basis} of $I$ (with respect to the given term order).

Gr\"obner bases have a tight connection with polynomial system solving. Namely, the solutions of a zero-dimensional  polynomial system $f_i(x_1,\ldots,x_n) = 0$ ($i = 1,\ldots,m$) can be read off from a lexicographic Gr\"obner basis of the corresponding ideal $(f_1,\dots,f_m)$ in $\mathbb{F}_q[x_1,\dots,x_n]$ thanks to the Shape Lemma and its extensions (see e.g. \cite{Caminata-Gorla2}).

There are several methods to compute Gr\"obner bases in practice. The oldest one is Buchberger's algorithm \cite{Buchberger}. Some of the fastest methods in use today are the \textit{linear-algebra algorithms}, such as F4 \cite{FaugereF4}, F5 \cite{FaugereF5}, and XL \cite{XLAlgorithm}. The core idea of these algorithms is to compute the Gr\"obner basis by doing Gaussian reduction on certain matrices, called \textit{Macaulay matrices}.
Let $d \in \mathbb{Z}_{\geq 1}$, and let $\T_{\leq d}$ be the set of terms in $\mathbb{F}_q[x_1,\ldots,x_n]$ of degree less than or equal to $d$. The \textit{Macaulay matrix} $\mathcal{M}_{\leq d}$ of $\mathcal{F}=\{f_1,\ldots,f_m\}$ is a matrix whose rows are indexed by the polynomials $t_{k,h}f_k$, where $k = 1, \ldots, m$ and $t_{k,h}$ ranges through all terms in $\T_{\leq d}$ such that $\deg(t_{k,h}f_k) \leq d$, and columns are indexed by the terms in $\T_{\leq d}$. The $(i,j)$ entry of $\mathcal{M}_{\leq d}$ is the coefficient of the $j$-th term in the polynomial of the $i$-th row. When $d$ is large enough, performing Gaussian elimination on $\mathcal{M}_{\leq d}$ provides a Gr\"obner basis of $(\mathcal{F})$. The smallest $d$ such that this happens is called \emph{solving degree} of $\mathcal{F}$ and we denote it by $\SolvingDegree_{\leq}(\mathcal{F})$. We often omit the subscript $\leq$ when the term order is clear from the context.

In practice, the linear-algebra-based methods mentioned above make extensive use of this idea, carefully selecting suitable submatrices of the Macaulay matrices in order to minimize the number of zero entries and avoid unnecessary computations. Moreover, although a Gr\"obner basis with respect to the lexicographic order is ultimately required to extract the solutions of the system, it is often more efficient to first compute a Gr\"obner basis with respect to the degree-reverse lexicographic order, and then convert it to lexicographic form using the FGLM algorithm \cite{FGLM} or related techniques.
That being said, the main computational bottleneck in solving such systems is typically the Gaussian elimination performed on the largest Macaulay matrix encountered during the computation. The size of this matrix depends on the number of variables $n$, the number of polynomials $m$, and the solving degree $\SolvingDegree_\leq(\mathcal{F})$.
For this reason, it is crucial to determine, or at least accurately estimate, the solving degree of a polynomial system and understand how it scales with respect to the system’s parameters. However, computing the solving degree a priori is generally very difficult. To address this issue, several alternative invariants, often more tractable, have been introduced and related to the solving degree, including the degree of regularity \cite{BFS04}, the last fall degree \cite{Caminata-Gorla1, lastfalldeg}, the first fall degree \cite{firstfalldeg}, and the Castelnuovo-Mumford regularity \cite{Caminata-Gorla2}.

\section{Algebraic Modeling with Modular Polynomials}\label{section:modularpolynomials}

The first natural way to produce an algebraic modeling for the Supersingular Isogeny Problem is using modular polynomials. This was explored by Takahashi et al. \cite{Taketal}. We briefly recall this modeling, which will serve us as a comparison.

\subsection{Background on modular polynomials}
We fix a finite field $\mathbb{F}_q$ of characteristic $p$ and let $N>1$ be an integer coprime with $p$. Then, there exists a polynomial $\Phi_N \in \mathbb{Z}[X,Y]$, called \textit{$N$-th modular polynomial}, such that for all $j_1,j_2 \in \mathbb{F}_{q}$ it holds that $\Phi_N(j_1,j_2) = 0$ if and only if $j_1,j_2$ are the $j$-invariants of elliptic curves over $\mathbb{F}_{q}$ that are related by an isogeny of degree $N$ defined over $\overline{\mathbb{F}}_p$ (see e.g. \cite[\S11.9]{Husemoeller}).

In the following proposition we collect some useful properties of modular polynomials.

\begin{proposition}\label{properties_modular_poly}
Let $N>1$ be an integer. The modular polynomial $\Phi_{N}(X,Y) \in \mathbb{Z}[X,Y]$ satisfies the following properties:
\begin{enumerate}
    \item $\Phi_{N}(X,Y) = \Phi_{N}(Y,X);$
    \item $\Phi_{N}$ is monic and it has degree $N + 1$ both in $X$ and $Y$;
    \item If $N$ is prime, then the highest degree form $\Phi_{N}^{\mathrm{top}}$ of $\Phi_{N}$ is $- X^N Y^N$.
\end{enumerate}
\begin{proof}
The first two properties are well known (see e.g. \cite[Lectures~19~and~20]{Sutherland}). We address the third. Let $N$ be prime, by combining \cite[Theorem~19.14~and~Lemma~20.9]{Sutherland} we get that $\Phi_N(X,Y)$ is monic in $X$ and in $Y$ and the top degree part of $\Phi_N(X,X)$ is $-X^{2N}$. Therefore, we can write
$$\Phi_N(X,Y) = X^{N+1} + Y^{N+1} - X^NY^N + \{\mbox{monomials of degree}< 2N\}.$$
This tells us exactly that $\Phi_N(X,Y)^{\mathrm{top}} = -X^NY^N$.
\end{proof}
\end{proposition}

\begin{example}
	The first two modular polynomials are the following.
\begin{equation}\label{2-3-mod_poly}
	\begin{aligned}
		\Phi_2(X,Y) &= X^3 + Y^3 - X^2Y^2 + 1488(X^2Y + XY^2) - 162000(X^2 + Y^2) +\\
		& + 40773375XY + 8748000000(X + Y) - 157464000000000;\\
		\Phi_3(X,Y)& = (X + Y)^4 - X^3Y^3 + 2232X^2Y^2(X + Y) + 36864000(X + Y)^3 + \\
		& - 1069960XY(X + Y)^2 + 2590058000X^2Y^2 + \\
		& + 8900112384000XY(X + Y) + 452984832000000(X + Y)^2 +\\
		& - 771751936000000000XY + 1855425871872000000000(X + Y).
	\end{aligned}
\end{equation}
Higher degree modular polynomials are larger. However, recall that they can be constructed algorithmically, for example by using the procedure outlined in \cite{Charles}.
\end{example}

One subtlety is that the vanishing of $\Phi_N(j_1,j_2)$ is equivalent to the existence of an isogeny which is defined over the algebraic closure $\overline{\mathbb{F}}_q$, while we are interested in working over a finite field. In our setup, this is not a problem. We address it in the following remark.

\begin{remark}
	First of all, we recall that every supersingular elliptic curve has $j$-invariant defined over $\mathbb{F}_{p^2}$ and it is isomorphic (over $\overline{\mathbb{F}}_p$) to an elliptic curve defined over $\mathbb{F}_{p^2}$ (\cite[Remark 9.7]{Silverman}).
	Now, if $E_1$ and $E_2$ are two supersingular elliptic curves defined over $\mathbb{F}_{p^2}$ that are connected via an isogeny of degree $N$ over $\overline{\mathbb{F}}_p$, then there exist two elliptic curves $E_1'$, $E_2'$, defined over $\mathbb{F}_{p^2}$, such that $E_1 \cong E_1'$ and $E_2 \cong E_2'$ over $\overline{\mathbb{F}}_p$ and there exists a $N$-isogeny between $E_1'$ and $E_2'$ defined over $\mathbb{F}_{p^2}$ (see \cite[Lemma 5.2]{JKPRST2018}).  Therefore, when considering chain of isogenies up to taking those isomorphic curves at the beginning and at the end, we may restrict without loss of generality to work only with supersingular curves and isogenies defined over $\mathbb{F}_{p^2}$.
\end{remark}

\subsection{The modeling}
Let $\ell\neq p$ be a prime number.
Given two supersingular elliptic curves $E_{\start}, E_{\finish}$ over $\mathbb{F}_{p^2}$ with  $j$-invariants $j_{\start}$ and $j_{\finish}$  that are connected via an isogeny $\varphi$ of degree $\ell^m$ ($m \in \mathbb{Z}_{\geq 1}$), we want to find the $j$-invariants of the curves in the path that connects them, i.e, the $j$-invariants of curves $E_1,\ldots,E_{m-1}$ such that

\begin{equation}\label{eq:chainisogenies}
E_{\start}\xrightarrow{\varphi_1} E_1\xrightarrow{\varphi_2}E_2\rightarrow\cdots\rightarrow E_{m-1}\xrightarrow{\varphi_m}  E_{\finish},
\end{equation}
where the $\varphi_i$'s are isogenies of degree $\ell$ such that
$$\varphi_m \circ \varphi_{m-1} \circ \cdots \circ \varphi_1 = \varphi.$$
Now, we can set up the following polynomial system.

\begin{tcolorbox}[
  title=\hypertarget{system1}{Polynomial system 1: Modular Polynomials},
  colframe=black,         
  colback=white,          
  coltitle=black,         
  colbacktitle=white,     
  boxrule=0.5pt,
  arc=4pt,
  fonttitle=\bfseries,
  sharp corners,
  top=1mm,
  bottom=0mm,
]
\raggedright $S_{m-1} = \mathbb{F}_{p^2}[j_1,\ldots,j_{m-1}]$\\
$\mathcal{M}_{\ell,m} = [\Phi_{\ell}(j_{\start},j_1),\Phi_{\ell}(j_1,j_2),\ldots,\Phi_{\ell}(j_{m-2},j_{m-1}),\Phi_{\ell}(j_{m-1},j_{\finish})] \subseteq S_{m-1}$
\begin{equation*}
\left\{
\begin{array}{l}
\Phi_{\ell}(j_{\start},j_1) = 0 \\
\Phi_{\ell}(j_1,j_2) = 0 \\
\cdots \cdots\\
\Phi_{\ell}(j_{m-2},j_{m-1}) = 0\\
\Phi_{\ell}(j_{m-1},j_{\finish}) = 0
\end{array}
\right.
\end{equation*}
\end{tcolorbox}

The properties of modular polynomials immediately give us the following result.

\begin{theorem}\label{theorem_modularpolynomials}
Let $E_{\start},E_{\finish}$ be two supersingular elliptic curves defined over $\mathbb{F}_{p^2}$ which are connected by an $\ell^m$-isogeny $\varphi$, with $\ell$ a prime number and $m \in \mathbb{Z}_{>1}$. Let $(\overline{j_1},\dots,\overline{j_{m-1}})\in\mathbb{F}_{p^2}^{m-1}$ be a solution of  Polynomial system~\hyperlink{system1}{1}, then $\overline{j_1},\dots,\overline{j_{m-1}}$ are the $j$-invariants of supersingular elliptic curves in a chain of isogenies of degree $\ell$ that connects  $E_{\start}$ and $E_{\finish}$ as in \eqref{eq:chainisogenies}.
\end{theorem}

Theorem~\ref{theorem_modularpolynomials} tells us that finding a solution of Polynomial system \hyperlink{system1}{1} allows us to reduce an instance of the Supersingular Isogeny Problem (SIP) of degree $\ell^m$ between  $E_{\start}$ and $E_{\finish}$ to $m$ instances of the SIP of degree $\ell$. When $\ell$ is small (such as $\ell=2,3$) the isogenies $\varphi_i$'s are easy to find. In fact, we point out that up to  $\ell \leq 31$, there exist SageMath libraries \cite{sagemath, Tsukazaki} that allow the efficient writing of all the $\ell$-isogenies of the decomposition. 

We postpone to Section~\ref{section:algebraicanalysis} and Section~\ref{section:experiments} an analysis of the algebraic properties of Polynomial system~\hyperlink{system1}{1} and a report on some experiments that we performed for $\ell=2,3$.

\section{Algebraic Modeling with Renes Polynomials in degree $2$}\label{section:Renes2}

In this section, we consider chain of isogenies of degree $2$. Given a supersingular elliptic curve $E$ we have three isogenies of degree $2$ with domain the curve $E$, indeed the $2$-isogeny graph is $3$-regular \cite{Kohel}. The kernel of a $2$-isogeny is a subgroup of order $2$. In other words, we have a degree $2$ isogeny associated to each non trivial element of $E[2]$. It turns out that when $E$ is in Montgomery form, these elements and the corresponding isogenies can be explicitly written thanks to Renes formulas \cite{Burdges-DeFeo, Costello-Hisil, Renes}. We use these formulas to construct polynomials and build an algebraic modeling for SIP.

\subsection{Montgomery curves and Renes formulas}\label{section:Renesformulas}
A supersingular elliptic curve $E$ over $\mathbb{F}_{p^2}$ is written in \textit{Montgomery form} if
$$E\colon By^2 = x^3 + A x^2 + x,$$
with $A,B \in \mathbb{F}_{p^2}$. It is clear that a nonsingular curve in Montgomery form is elliptic, while not every elliptic curve in Weierstrass form can be transformed into a Montgomery form. A necessary and sufficient condition is given in the following theorem from \cite{OKS00}. 

\begin{theorem}[\cite{OKS00}]
Let $E \colon y^2 = x^3 + a x+ b $ be an elliptic curve defined over $\mathbb{F}_{q}$. Then $E$ is isomorphic to a Montgomery curve $E' \colon By^2 = x^3 + Ax^2 + x$ over $\mathbb{F}_{q}$ if and only if
\begin{enumerate}
    \item $E$ has an $\mathbb{F}_{q}$-rational $2$-torsion point $(\alpha,0)$;
    \item $3 \alpha^2 + a = s^2$ for some $s \in \mathbb{F}_{q} \smallsetminus \{0\}$. 
\end{enumerate}
In this case, the coefficients of $E'$ are $A = 3 \alpha s^{-1}$ and $B = s^{-1}$.
\end{theorem}

\begin{remark}
If $E\colon B y^2 = x^3 + Ax^2 + x$ is an elliptic curve in Montgomery form over $\mathbb{F}_{p^2}$, then it is isomorphic (over $\mathbb{F}_{p^4}$) to the curve
$E' \colon y^2 = x^3 + Ax^2 + x$, which is a quadratic twist of $E$.
Thus, for simplicity in what follows we will restrict to curves in Montgomery form with  $B = 1$ and  $A \in \mathbb{F}_{p^2}$.
\end{remark}
We consider a supersingular elliptic curve $E$ written in Montgomery form as
\begin{equation}\label{eq:montgomery}
E:y^2 = x^3 + Ax^2 + x.	
	\end{equation}
The $2$-torsion points of $E$ are
\[
E[2]=\{O_E,(0,0), (\alpha,0),(\alpha^{-1},0)\},
\]
where $\alpha \in \mathbb{F}_{p^4}$ is a root of $x^2 + Ax + 1$, and $A = - \alpha - \alpha^{-1}$.
For curves with $j$-invariant $\neq0,1728$ the roots $\alpha,\alpha^{-1}$ of $x^2 + Ax + 1$ are indeed in $\mathbb{F}_{p^2}$.

\begin{lemma}\label{lem:2torsion-rationality}
Let $E$ be a supersingular elliptic curve defined over $\mathbb{F}_{p^2}$ such that $|E(\mathbb{F}_{p^2})| = (p \pm 1)^2$. Then $E[2] \subseteq E(\mathbb{F}_{p^2})$, i.e., the $2$-torsion points of $E$ are $\mathbb{F}_{p^2}$-rational. 
\end{lemma}
\begin{proof}
Thanks to \cite[Lemma 4.8 (ii)]{Schoof}, we have that $E(\mathbb{F}_{p^2}) \cong \mathbb{Z}/(p+1)\mathbb{Z} \times \mathbb{Z}/(p+1)\mathbb{Z}$ or $E(\mathbb{F}_{p^2}) \cong \mathbb{Z}/(p-1)\mathbb{Z} \times \mathbb{Z}/(p-1)\mathbb{Z}$. On the other hand, we have that $E[2] \cong \mathbb{Z}/2\mathbb{Z} \times \mathbb{Z}/2\mathbb{Z}$. Since $2 \mid p \pm 1$, it holds that $\mathbb{Z}/(p \pm 1)\mathbb{Z}$ has a subgroup of order $2$, and it must be $\mathbb{Z}/2\mathbb{Z}$. So we must have that $E[2] \subseteq E(\mathbb{F}_{p^2})$. 
\end{proof}

\begin{remark}\label{rem:numofratpointsperjinv}
	The condition $|E(\mathbb{F}_{p^2})| = (p \pm 1)^2$ in Lemma~\ref{lem:2torsion-rationality} may seem restrictive, however it is satisfied by any supersingular elliptic curve $E$ with $j(E) \not= 0,1728$.
	Indeed, let $t=p^2+1-|E(\mathbb{F}_{p^2})|$ be the trace of the Frobenius endomorphism $\phi_{p^2}$ on $E$. Since $E$ is supersingular, by Hasse's bound we have that $t\in\{0,\pm p,\pm 2p\}$. If $t=0$, then the characteristic polynomial of the Frobenius $\phi_{p^2}$ is $h(x)=x^2+p^2$. Since $h(x)$ has a root $\phi_{p^2}$ in $\mathrm{End}(E)$, its discriminant $\Delta=-4p^2$ must be a square.  This implies that $E$ has a nontrivial automorphism of order $4$. Since the characteristic of the field is $\neq2,3$, this happens only if $j(E)=1728$ \cite[Theorem~10.1]{Silverman}. With a similar argument one can show that  $t=\pm p$ forces $j(E)=0$. Thus, if  $j(E) \not= 0,1728$, we have $t=\pm 2p$ and $|E(\mathbb{F}_{p^2})| = (p \pm 1)^2$ (see also \cite{ADJ2019}).
\end{remark}

Now, we consider a supersingular elliptic curve $E$ over $\mathbb{F}_{p^2}$ written in Montgomery form as in~\eqref{eq:montgomery} with  $|E(\mathbb{F}_{p^2})| = (p \pm 1)^2$. By Lemma~\ref{lem:2torsion-rationality},  the $2$-torsion points of $E$ are in $\mathbb{F}_{p^2}$. Thus, we can factor the polynomial $x^2+Ax+1$ over $\mathbb{F}_{p^2}$ and write $E$ in \textit{simplified Montgomery form} \cite{Burdges-DeFeo}:
$$E\colon y^2 = x(x - \alpha)(x - \alpha^{-1}),$$
with $\alpha \in \mathbb{F}_{p^2}$.
The subgroups of order $2$ $\langle(\alpha,0)\rangle$ and $\langle(\alpha^{-1},0)\rangle$ are the kernels of two distinct $2$-isogenies $\varphi_1$ and $\varphi_2$ with domain $E$. It turns out that the codomain elliptic curves of $\varphi_1$ and $\varphi_2$ can also be written in simplified Montgomery form.

The following lemma due to Renes \cite[Proposition 2]{Renes} formalizes our claims.

\begin{lemma}[Renes formulas, degree 2]\label{lemma:Renesformulas}
Let $E$ be a supersingular elliptic curve over $\mathbb{F}_{p^2}$ written in simplified Montgomery form $y^2 = x(x - \alpha)(x - \alpha^{-1})$. Then, there exist two $2$-isogenies $\varphi_1 \colon E \rightarrow E_1$, $\varphi_2 \colon E \rightarrow E_2$ defined over $\mathbb{F}_{p^2}$ such that $E_1$, $E_2$ are defined over $\mathbb{F}_{p^2}$ and can be written in Montgomery form as
$$E_1\colon y^2 = x^3 + A_1 x^2 + x,\quad E_2\colon y^2 = x^3 + A_2 x^2 + x,$$
where $A_1 = 2 - 4 \alpha^2$ and  $A_2 = 2 - 4(\alpha^{-1})^2$.
\end{lemma}

\begin{remark}
The remaining $2$-isogeny with domain $E$ is the one with kernel $\langle(0,0)\rangle$. We will not keep track of this isogeny in the algebraic modeling we are going to construct. This is because when we have a chain of $2$-isogenies, the isogeny with kernel $\langle(0,0)\rangle$ (apart from the first step of the chain) is the backtracking isogeny, i.e., the dual isogeny to the previous step in the chain (see \cite[Corollary 1]{Renes}).
\end{remark}

\subsection{Renes polynomials}

We are going to use Renes formulas to construct polynomials which we will use in our modeling. First, we define them and then we explain their properties.

\begin{definition}\label{def:Renespolynomials2}
The \emph{Renes polynomials of degree $2$} are the following polynomials in $\mathbb{Z}[X,Y]$:
\[
\begin{split}
	\Psi_1(X,Y) &= XY + Y^2 + 1; \\
	\Psi_2(X,Y) &= -4X^2Y + Y^2 + 2Y + 1; \\
	\Psi_3(X,Y) &= 65536X^{12} - 196608X^{10} + 208896X^8 - 90112X^6 - X^4Y + 13056X^4 +\\
	&+  X^2Y - 768X^2 + 16.
\end{split}
\]	
\end{definition}

\begin{proposition}\label{prop:Renespolynomialsdeg2}
Let $E$ be a supersingular elliptic curve over $\mathbb{F}_{p^2}$ such that $|E(\mathbb{F}_{p^2})| = (p \pm 1)^2$ written in Montgomery form $E:y^2 = x^3 + Ax^2 + x$. Let  $\varphi_1 \colon E \rightarrow E_1$, $\varphi_2 \colon E \rightarrow E_2$ be two degree $2$ isogenies (defined over $\mathbb{F}_{p^2}$) with kernel $\neq\langle(0,0)\rangle$. The following facts hold.
\begin{enumerate}
	\item Let $\alpha,\alpha^{-1}\in\mathbb{F}_{p^2}$ be the roots of $\Psi_1(A,Y)$. Then, the simplified Montgomery form of $E$ is $y^2 = x(x - \alpha)(x - \alpha^{-1})$.
	\item Let $\gamma_1,\gamma_2,\gamma_3,\gamma_4$ be the four roots of $\Psi_2(\alpha,Y)$ and $\Psi_2(\alpha^{-1},Y)$. Then, up to relabeling, $\gamma_3 = \gamma_1^{-1}$, $\gamma_4 = \gamma_2^{-1}$ and the simplified Montgomery forms of $E_1$ and $E_2$ are 
	$$E_i \colon y^2 = x(x- \gamma_i)(x- \gamma_i^{-1}).$$
\end{enumerate}
\end{proposition}

\begin{proof}
	\begin{enumerate}
		\item This is clear from the discussion in \S\ref{section:Renesformulas}. Notice that $\alpha,\alpha^{-1}\in\mathbb{F}_{p^2}$ by Lemma~\ref{lem:2torsion-rationality}.
		\item Let $i\in\{1,2\}$. By Lemma~\ref{lemma:Renesformulas}, also $E_i$ can be written in Montgomery form. Moreover, Tate's Theorem (\cite[Section 3 Theorem 1]{Tate}) yields $|E_i(\mathbb{F}_q)| = |E(\mathbb{F}_q)|= (p \pm 1)^2$. Thus, by Lemma~\ref{lem:2torsion-rationality} $E_i$ can be written in simplified Montgomery form
		$$E_i \colon y^2 = x(x- \gamma)(x- \gamma^{-1}),$$
		where $\gamma$ is a root of $x^2 + (2 - 4\beta^2)x + 1$ with $\beta \in \{\alpha,\alpha^{-1}\}$.
		 So, we have
		\begin{equation}\label{eq:gammai}
			\gamma_i = - 1 + 2\beta^2 + c_{i}, \ \ i=1,2,  
		\end{equation}
		where the $c_i$'s are the two square roots of the discriminant $4\beta^4 - 4\beta^2$. If we isolate $c_i$ in \eqref{eq:gammai} and square both sides we get that $\gamma$ must satisfy
		\begin{equation}\label{eq:Renform}
			-4 \beta^2 \gamma + \gamma^2 + 2\gamma + 1 = 0,
		\end{equation}
		that is $\gamma$ is a root of $\Psi_2(\beta,Y)$.
		\end{enumerate}
        \end{proof}

       \begin{proposition}\label{prop:propertiesofPsi3}
        Let $E$ be a supersingular elliptic curve over $\mathbb{F}_{p^2}$ in Montgomery form. 
           Let $\delta\in\mathbb{F}_{p^2}$ be a root of $\Psi_3(X,j(E))$. Then
there exists an isogeny of degree $2$ (defined over $\mathbb{F}_{p^2}$)  from $y^2=x(x-\delta)(x-\delta^{-1})$ to an elliptic curve $y^2 = x^3 + (2 - 4\delta^2)x^2 + x$ with $j$-invariant $j(E)$.
       \end{proposition} 
       \begin{proof} 
        We recall that we have six possible elliptic curves in Montgomery form with the same $j$-invariant of $E$. Indeed, their $A$'s are given by the relation 
		$$j(E) = 256\frac{(A^2 - 3)^3}{A^2 - 4}.$$
        Clearing the denominator in the previous equation we get that the roots of the polynomial $f(Z)=(Z^2 - 4)j(E) - 256(Z^2 - 3)^3$ are precisely these six possible $A$'s. Now, let $g(X,Z)=2 - 4X^2 - Z$ be the polynomial obtained from Renes formulas (Lemma~\ref{lemma:Renesformulas}). In particular, for a fixed value $\delta$, the roots of $g(X,\delta)$ are precisely the $A$-coefficients of the Montgomery forms of the codomain elliptic curves of the $2$-isogenies from  $y^2=x(x-\delta)(x-\delta^{-1})$.
        We consider the elimination ideal $I=(g(X,Z),f(Z))\cap\mathbb{Q}[X]$. A computation shows that $I$ is generated by the polynomial
        \begin{align*}
			& 65536X^{12} - 196608X^{10} + 208896X^8 - 90112X^6 - X^4 j(E) + 13056X^4 +\\
			&+  X^2 j(E) - 768X^2 + 16 \in \mathbb{Z}[X].
		\end{align*}
        This is exactly $\Psi_3(X,j(E))$. Alternatively, one can check that $\Psi_3(X,j(E)) = -\frac{1}{16}f(2 - 4X^2)$.
        
Now, let $\delta \in \mathbb{F}_{p^2}$ be a root of $\Psi_3(X,j(E))$. By construction, $\delta$ satisfies $f(2- 4\delta^2) = 0$.
       Thus,  $A'=2 - 4\delta^2$ is one of the six possible Montgomery coefficients $A$'s coming from $j(E)$. From what is said above, we know that there exists an isogeny of degree 2 (defined over $\mathbb{F}_{p^2}$) from $y^2=x(x-\delta)(x-\delta^{-1})$ to the elliptic curve $y^2 = x^3 + A'x^2 + x$ with $j$-invariant $j(E)$.
	\end{proof}

\begin{remark}
Notice that the curve $y^2 = x(x-\delta)(x-\delta^{-1})$ of Proposition~\ref{prop:propertiesofPsi3} is not an elliptic curve iff $\delta \in \{0,\pm 1\}$. Indeed, if $\delta \in \{0,\pm 1\}$ then $A' \in \{\pm 2\}$ which implies that $j(E)$ is not well-defined.
\end{remark}

\subsection{The modeling}

We assume the following setup. We have two supersingular elliptic curves $E_{\start}$ and $E_{\finish}$ over $\mathbb{F}_{p^2}$ in Montgomery form (with $B = 1$) that are connected via an isogeny $\varphi$ of degree $2^m$ ($m \in \mathbb{Z}_{\geq 1}$). Moreover, suppose $|E_{\start}(\mathbb{F}_{p^2})| = (p \pm 1)^2$, thus by Lemma~\ref{lem:2torsion-rationality}, $E_{\start}$ can be put in simplified Montgomery form, and by Tate's Theorem the same applies to $E_{\finish}$ and all elliptic curves in the chain.
Finally, we assume that the first $2$-isogeny (the one with domain $E_{\start}$) in the decomposition of $\varphi$ is not the one associated with the $2$-torsion point $(0,0)$. Let $A_{\start}, A_{\finish} \in \mathbb{F}_{p^2}$ be their Montgomery $x^2$-coefficients and $j(E_{\finish})$ the $j$-invariant of $E_{\finish}$ computed from $A_{\finish}$. Using Renes polynomials, we construct a polynomial system whose solutions correspond to all possible degree $2^m$ isogenies from  $E_{\start}$ to an elliptic curve with $j$-invariant $j(E_{\finish})$.

To make things more precise, we introduce the following definition which will be important also for the degree $3$ case.
\begin{definition}\label{def:P(d,m,E,j)set}
	Let $d\geq2$, $m\geq1$ be integers and let $E_{\start}$ and $E_{\finish}$ be two supersingular elliptic curves which are connected via an isogeny of degree $d^m$. We define the set
	\begin{equation}\label{def:PdmEj}
		\begin{aligned}
			P(d,m,E_{\start},j(E_{\finish})) = \{&E_{\start} \xrightarrow[]{\varphi_0} E_1 \xrightarrow[]{\varphi_1}  \cdots \xrightarrow[]{\varphi_{m-1}}  E_m \ | \text{ non-backtracking, } \\ &\ker\varphi_0\neq\langle(0,0)\rangle, \ \deg(\varphi_i) = d \ \ \forall i,\ j(E_m) = j(E_{\finish})\}/\sim 
		\end{aligned}   
	\end{equation}
	where $E_{\start} \xrightarrow[]{\varphi_0} E_1 \xrightarrow[]{\varphi_1}  \cdots \xrightarrow[]{\varphi_{m-1}}  E_m$ and $E_{\start} \xrightarrow[]{\varphi_0'} E_1' \xrightarrow[]{\varphi_1'}  \cdots \xrightarrow[]{\varphi_{m-1}'}  E_m'$ are equivalent under $\sim$ if and only if for $i=1,\ldots,m$ there exist isomorphisms $\psi_i:E_i\xrightarrow{\cong}E_i'$  such that the following diagram commutes
    \[
\begin{tikzcd}
E_{\start} \arrow[r, "\varphi_0"] \arrow[d, "\mathrm{id}"]
& E_1  \arrow[d, "\psi_1"] \arrow[r]
& \cdots \arrow[r]
& E_{m-1} \arrow[d, "\psi_{m-1}"] \arrow[r, "\varphi_{m-1}"] 
& E_{m} \arrow[d, "\psi_{m}"]
\\
E_{\start} \arrow[r, "\varphi'_0"] 
& E_1' \arrow[r]
& \cdots \arrow[r]
& E'_{m-1}\arrow[r, "\varphi'_{m-1}"] 
& E'_{m}
\end{tikzcd}
    \]
In particular, we have $j(E_i) = j(E_i')$ and $\mathrm{ker}(\varphi_i)=\mathrm{ker}(\varphi_i'\circ \psi_i)$ for all $i=1,\dots, m$.
\end{definition}

The set $P(d,m,E_{\start},j(E_{\finish}))$ parametrizes all possible non-backtracking paths of degree $d$ isogenies of length $m$ that start from $E_{\start}$ and end into a supersingular elliptic curve with the same $j$-invariant as $E_{\finish}$. 
We recall that by non-backtracking we mean that we do not allow two consecutive isogenies in the path to be dual to each other.
The additional condition that the kernel of the first isogeny is not generated by the point $(0,0)$ guarantees that when the curve $E_{\start}$ is in Montgomery form and $d=2$, then these paths can be obtained from the solutions of the following polynomial system build up from Renes polynomials.

\begin{tcolorbox}[
        title=\hypertarget{system2}{Polynomial system 2: Renes Formulas for $2$-isogenies, complete},
	colframe=black,         
	colback=white,          
	coltitle=black,         
	colbacktitle=white,     
	boxrule=0.5pt,
	arc=4pt,
	fonttitle=\bfseries,
	sharp corners,
	top=0mm,
	bottom=2mm,
	]
	$R_m = \mathbb{F}_{p^2}[\alpha_0,\ldots,\alpha_{m-1}]$\\
        $\mathcal{F}_{2\text{-Renes},m} = [\Psi_1(A_{\start},\alpha_{0}),\Psi_2(\alpha_0,\alpha_1),\ldots,\Psi_2(\alpha_{m-2},\alpha_{m-1}), \Psi_3(\alpha_{m-1},j(E_{\finish}))] \subseteq R_{m}$
		\begin{equation*}\label{eq_syst_Renes_notat}
				\begin{cases}
					&\Psi_1(A_{\start},\alpha_0) = 0 \\
					&\Psi_2(\alpha_0,\alpha_1) = 0 \\
					&\Psi_2(\alpha_1,\alpha_2) = 0 \\
					&\cdots \\
					&\Psi_2(\alpha_{m-2},\alpha_{m-1}) = 0 \\
					&\Psi_3(\alpha_{m-1},j(E_{\finish})) = 0.
				\end{cases}    
			\end{equation*}
\end{tcolorbox}

\begin{theorem}\label{thm:fund_thm_deg2}
	Let $E_{\start},E_{\finish}$ be two supersingular elliptic curves in Montgomery form over $\mathbb{F}_{p^2}$ (with the $y^2$-coefficient equal to 1). Suppose that $|E_{\start}(\mathbb{F}_{p^2})| = (p \pm 1)^2$.
    Then, given $m \in \mathbb{Z}_{\geq 1}$, there is a bijective map 
    \[\Theta_2:V_{\mathbb{F}_{p^2}}\left(\mathcal{F}_{2\text{-Renes},m}\right)\longrightarrow P(2,m,E_{\start},j(E_{\finish})),\]
    where $V_{\mathbb{F}_{p^2}}\left(\mathcal{F}_{2\text{-Renes},m}\right)$
    is the set of $\mathbb{F}_{p^2}$-rational solutions to Polynomial system~\hyperlink{system2}{2}.
\end{theorem}
\begin{proof}
	First, let $(\alpha_0,\ldots,\alpha_{m-1}) \in {(\mathbb{F}_{p^2})}^{m}$ be a solution of Polynomial system \hyperlink{system2}{2}. Then, by Proposition~\ref{prop:Renespolynomialsdeg2},  the simplified Montgomery form of $E_{\start}$ is $y^2 = x(x - \alpha_0)(x - \alpha_0^{-1})$ and we have a degree~$2$ isogeny $\varphi_0:E_{\start}\rightarrow E_1$, with $E_1:y^2 = x(x- \alpha_1)(x- \alpha_1^{-1})$ and $\ker\varphi_0\neq\langle(0,0)\rangle$. Applying Proposition~\ref{prop:Renespolynomialsdeg2} to $E_1$ and then recursively, we obtain a chain $E_{\start} \overset{\varphi_0}{\longrightarrow} E_1 \overset{\varphi_1}{\longrightarrow} \cdots \overset{\varphi_{m-1}}{\longrightarrow} E_m$, where each curve in the path has simplified Montgomery form $y^2 = x(x - \alpha_i)(x - \alpha_i^{-1})$, $\deg\varphi_i=2$, and $\ker\varphi_i\neq\langle(0,0)\rangle$. The last condition guarantees that the path is non-backtracking. Also, by Proposition~\ref{prop:propertiesofPsi3}, we obtain that $j(E_m) = j(E_{\finish})$. Thus, we have constructed an element of the set $P(2,m,E_{\start},j(E_{\finish}))$. This construction defines the desired map $\Theta_2$.
    
	Now, we construct an inverse map $\Lambda_2$ in the opposite direction.
	Suppose that we have a path $\varphi$ of $m$ degree $2$ non-backtracking isogenies between $E_{\start}$ and $E_{m}$ with $j(E_{\finish})=j(E_m)$ and $\ker\varphi_0\neq\langle(0,0)\rangle$:
    $$\varphi: E_{\start} \xrightarrow[]{\varphi_0} E_1 \xrightarrow[]{\varphi_1} \cdots \xrightarrow[]{\varphi_{m-1}} E_m.$$
We select an appropriate representative in the equivalence class of $\varphi$ in $P(2,m,E_{\start},j(E_{\finish}))$ and we use this representative to construct an element $(\alpha_0,\dots,\alpha_{m-1})\in V_{\mathbb{F}_{p^2}}\left(\mathcal{F}_{2\text{-Renes},m}\right)$.    
    Since $E_{\start}$ is in Montgomery form, then $\ker(\varphi_0)$ is generated by a point $(\alpha_0,0)$, where $\alpha_0$ is one of the two distinct (since $E_{\start}$ is an elliptic curve) solutions of the equation $\Psi_1(A_{\start},Y) = Y^2 + A_{\start}Y + 1 = 0$. Notice that $\alpha_0\in\mathbb{F}_{p^2}$ by Lemma~\ref{lem:2torsion-rationality}. Thanks to Renes formulas (Lemma~\ref{lemma:Renesformulas}), each $2$-torsion point $(\alpha,0)$ defines a distinct $2$-isogeny from $E_{\start}$ to an elliptic curve $E(\alpha)$ in Montgomery form defined by the coefficient $A = 2 - 4\alpha^2$. In particular, the codomain curve $E(\alpha_0)$ is isomorphic over $\overline{\mathbb{F}}_p$ to the curve $E_1$ in the fixed path $\varphi$, thus $j(E_1) = j(E(\alpha_0))$.
    Therefore, by composing with an isomorphism $\psi_1:E_1\rightarrow E(\alpha_0)$ we may replace $E_1$ by $E(\alpha_0)$, which is in Montgomery form, without changing the equivalence class of the path $\varphi$. Since the path is non-backtracking, the kernel of $\varphi_1$ is generated by a point $(\alpha_1,0)$, where by Proposition~\ref{prop:Renespolynomialsdeg2} $\alpha_1$ is a solution to the equation $\Psi_2(\alpha_0,Y) = Y^2 + (2 - 4\alpha_0^2)Y + 1 = 0$. 
    We observe that $\alpha_1\in\mathbb{F}_{p^2}$ since $|E(\alpha_0)(\mathbb{F}_{p^2})| = (p \pm 1)^2$ by Tate's theorem.
    As before, we may compose with an isomorphism $\psi_2:E_2\rightarrow E(\alpha_1)$ and replace $E_2$ with a curve in Montgomery form $E(\alpha_1)$ without changing the representative class of the path $\varphi$. Proceeding in this way, we obtain elements $\alpha_0,\ldots,\alpha_{m-1}\in\mathbb{F}_{p^2}$ which satisfy the equations involving the polynomials $\Psi_1$ and $\Psi_2$ of Polynomial system~\hyperlink{system2}{2}. Notice that, since the path is non-backtracking, at each step we do not have to consider the isogeny with kernel generated by $(0,0)$. At the last step, from the condition $j(E_{\finish}) = j(E(\alpha_{m-1}))$, we also obtain that $\alpha_{m-1}$ is a solution to the equation $\Psi_3(X,j(E_{\finish})) = 0$ by Proposition~\ref{prop:propertiesofPsi3}. Therefore, $(\alpha_0,\dots,\alpha_{m-1})$ is a solution of Polynomial system~\hyperlink{system2}{2}. In particular, this construction defines a map $\Lambda_2:P(2,m,E_{\start},j(E_{\finish}))\longrightarrow V_{\mathbb{F}_{p^2}}\left(\mathcal{F}_{2\text{-Renes},m}\right)$ which is the inverse of the map $\Theta_2$ defined above, hence $\Theta_2$ is bijective.
\end{proof}

\begin{remark}[Comparison with modular polynomials]
We point out some distinctions between the approach with Renes polynomials and with modular polynomials.
\begin{enumerate}
	\item Once we have a solution of Polynomial system~\hyperlink{system2}{2}, we are able to directly recover the explicit curves (in simplified Montgomery form) that are in the path of the $2^m$-isogeny, while this is not immediate with modular polynomials. Moreover, we can also give a precise description of the $2$-isogenies in the path. Those formulas can be found in \cite[Proposition 2]{Renes}.
	\item  It is \underline{not} true in this setting that there exists a $2$-isogeny between two elliptic curves in simplified Montgomery form with parameters $\alpha_1$ and $\alpha_2$ if and only if $\Psi_2(\alpha_1,\alpha_2) = 0$. This is due to the fact that Renes polynomials do not take into account the different representations of the curves, which may have different Montgomery forms isomorphic to each other. Moreover, Renes polynomials do not consider the backtracking isogeny, associated to the $2$-torsion point $(0,0)$. Indeed,  $\Psi_2(X,Y)$ is not a symmetric polynomial.
\end{enumerate}  
\end{remark}

\subsection{The (compact) modeling}
From a computational point of view, the presence of the degree $12$ polynomial $\Psi_3$ in Polynomial System~\hyperlink{system2}{2} makes it harder to solve. To fix this issue, we consider $6$ systems where the last equation is replaced by $(2 - 4\alpha_{m-1}^2) - A = 0$, where $A$ is one of the $6$ possible roots of the polynomial that codifies all the possible $A$'s of a curve in Montgomery form with fixed $j$-invariant $j(E_{\finish})$. Explicitly, we precompute those $A$'s as the roots of $256(X^2 - 3)^3 - j(E_{\finish})(X^2 - 4) \in \mathbb{F}_{p^2}[X]$ and we plug them into the polynomial
$$\Psi_4 = -4X^2 - Y + 2 \in \mathbb{Z}[X].$$
We obtain $6$ equations $\Psi_4(\alpha_{m-1},A) = 0$, which give rise to $6$ new polynomial systems of the following form.

\begin{tcolorbox}[
        title=\hypertarget{system3}{Polynomial system 3: Renes Formulas for $2$-isogenies, compact},
	colframe=black,         
	colback=white,          
	coltitle=black,         
	colbacktitle=white,     
	boxrule=0.5pt,
	arc=4pt,
	fonttitle=\bfseries,
	sharp corners,
	top=0mm,
	bottom=2mm,
	]
	$R_m = \mathbb{F}_{p^2}[\alpha_0,\ldots,\alpha_{m-1}]$\\
        $\mathcal{R}_{2,m} = [\Psi_1(A_{\start},\alpha_{0}),\Psi_2(\alpha_0,\alpha_1),\ldots,\Psi_2(\alpha_{m-2},\alpha_{m-1}), \Psi_4(\alpha_{m-1},A)] \subseteq R_{m}$
		\begin{equation*}\label{eq_syst_Renes_notat}
				\begin{cases}
					&\Psi_1(A_{\start},\alpha_0) = 0 \\
					&\Psi_2(\alpha_0,\alpha_1) = 0 \\
					&\Psi_2(\alpha_1,\alpha_2) = 0 \\
					&\cdots \\
					&\Psi_2(\alpha_{m-2},\alpha_{m-1}) = 0 \\
					&\Psi_4(\alpha_{m-1},A) = 0.
				\end{cases}    
			\end{equation*}
\end{tcolorbox}
Now, the degree of the last polynomial is $2$ instead of $12$. Moreover, only one of these $6$ systems will have a solution, but the computations can be parallelized and, once a solution is found, all  other computations can be aborted.
Notice that, since we already have $A_{\finish}$ and we know that $-A_{\finish}$ is also an admissible $A$, to find the other $A$'s we can now find the roots of a polynomial of degree $4$. This means that the pre-computation of the $A$'s is computationally irrelevant with respect to the computation of a Gr\"obner basis.

For these reasons, we will focus on Polynomial System \hyperlink{system3}{3} rather than  Polynomial System~\hyperlink{system2}{2}. 
We postpone to Section~\ref{section:algebraicanalysis} and Section~\ref{section:experiments} an analysis of its algebraic properties and a report on some computational experiments.

\section{Algebraic Modeling with Renes Polynomials in degree $3$}\label{section:Renes3}
In this section, we consider chain of isogenies of degree $3$. Given a supersingular elliptic curve $E$ we have four isogenies of degree $3$ with domain the curve $E$, and the kernel of a $3$-isogeny is a subgroup of $E[3]$ of order $3$. In \cite{Renes}, Renes gives a description for degree $3$-isogenies similar to the degree $2$ in the previous subsection. We use it to construct polynomials and build an algebraic modeling for SIP.

\subsection{Triangular curves and Renes formulas}
A supersingular elliptic curve $E$ is in \emph{triangular form} if it is given by an equation 
\[
E \colon \ y^2+axy+by=x^3,
\]
for some $a,b\in \mathbb{F}_{p^2}$. If $b$ is a cube in $\mathbb{F}_{p^2}$, then writing $b=\beta^3$ for some $\beta\in\mathbb{F}_{p^2}$, we obtain $b=1$ through the transformation $(x,y)\mapsto(x/\beta^2,y/\beta^3)$. For simplicity, in the following, we will always assume that the elliptic curve admits a triangular form with $b=1$. We recall that in this case the condition of being non-singular is equivalent to $a^3\neq27$.

As for Montgomery curves, not every supersingular elliptic curve can be put into triangular form. However, if the curve $E/\mathbb{F}_{p^2}$ contains a rational point of order $3$, then it can be written in triangular form (see \cite[Theorem 5.2]{Bernstein-C-K-L}). We point out that when an elliptic curve is in triangular form, the point $(0,0)$ has order $3$. The following lemma, whose proof is analogous to the proof of Lemma~\ref{lem:2torsion-rationality}, gives a sufficient criterion for the rationality of points of order $3$.

\begin{lemma}\label{lem:3torsion-rationality}
Let $E$ be a supersingular elliptic curve defined over $\mathbb{F}_{p^2}$.  Assume that $|E(\mathbb{F}_{p^2})| = (p - 1)^2$ if $p\equiv 1\mod 3$ and  $|E(\mathbb{F}_{p^2})| = (p + 1)^2$ if $p\equiv 2\mod 3$. Then $E[3] \subseteq E(\mathbb{F}_{p^2})$, i.e., the $3$-torsion points of $E$ are $\mathbb{F}_{p^2}$-rational.
\end{lemma}

\begin{remark}
Notice that, if we are in a case for a curve $E$ as in Lemma~\ref{lem:3torsion-rationality} where $p\equiv 1\mod 3$ but $|E(\mathbb{F}_{p^2})| = (p+1)^2$ (resp. $p\equiv 2\mod 3$ but $|E(\mathbb{F}_{p^2})| = (p-1)^2$), we can pass to a twist $E_{qt}$ of $E$ such that $|E_{qt}(\mathbb{F}_{p^2})| = (p-1)^2$ (resp. $|E_{qt}(\mathbb{F}_{p^2})| = (p+1)^2$) and work with $E_{qt}$. Therefore, up to a twist, we can always apply Lemma \ref{lem:3torsion-rationality}.
\end{remark}

As for the degree $2$ case, if $E$ is a supersingular elliptic curve in triangular form and $\varphi:E\rightarrow E'$ is a degree $3$ isogeny with $\ker\varphi\neq\langle(0,0)\rangle$, then also $E'$ can be put in triangular form and the $\mathbb{F}_{p^2}$-rationality of the $3$-torsion points is preserved. More precisely, the $3$-isogenies with domain $E$ can be explicitly described by \emph{Renes formulas} as follows (cf. \cite[Proposition~4, Corollary 3]{Renes}).

\begin{lemma}[Renes formulas, degree 3]\label{lemma:Renes3}
Given a supersingular elliptic curve $E\colon y^2+axy+y=x^3$ in triangular form and $(x_P,y_P)\in E(\mathbb{F}_{p^2})$ a $3$-torsion point distinct from $(0,0)$, then we have a $3$-isogeny $\varphi:E\rightarrow E'$, where $E'\colon y^2+Axy+y=x^3$ with $A=-3(2+ax_P)$. Moreover, the dual isogeny $\hat{\varphi}:E'\rightarrow E$ is the one with kernel $\langle(0,0)\rangle$.
\end{lemma}

\subsection{Renes polynomials}

We are going to use Renes formulas to construct polynomials which we will use in our modeling. First, we define them and then we explain their properties.

\begin{definition}\label{def:Renesdegree3}
    The \emph{Renes polynomials of degree $3$} are the following polynomials in $\mathbb{Z}[X,Y]$:
    \[
    \begin{split}
     \Gamma_1(X,Y)&=X^3Y^2 + 3X^3Y + 9X^3 - Y^3 - 18Y^2 - 108Y - 216;\\
    \Gamma_2(X,Y)&= X^{12} - 72X^9 + 1728X^6 - X^3Y - 13824X^3 + 27Y.
    \end{split}
    \]
\end{definition}

\begin{proposition}\label{prop:properties of Gamma1}
Let $E$ be a supersingular elliptic curve over $\mathbb{F}_{p^2}$ written in triangular form as $y^2+axy+y=x^3$. Assume that $|E(\mathbb{F}_{p^2})| = (p - 1)^2$ if $p\equiv 1\mod 3$ and  $|E(\mathbb{F}_{p^2})| = (p + 1)^2$ if $p\equiv 2\mod 3$.
For $i\in\{1,2,3\}$, let  $\varphi_i \colon E \rightarrow E_i$ be the three degree $3$ isogenies (defined over $\mathbb{F}_{p^2}$) with kernel $\neq\langle(0,0)\rangle$.
	Let $\gamma_1,\gamma_2,\gamma_3\in\mathbb{F}_{p^2}$ be the roots of $\Gamma_1(a,Y)$. Then, (up to relabeling) the triangular form of $E_i$ is $y^2+\gamma_i xy+y=x^3$.
\end{proposition}
\begin{proof}
The assumption on the number of rational points of $E$ guarantees that the $3$-torsion points of $E$ are rational by Lemma~\ref{lem:3torsion-rationality}. Then, by Lemma~\ref{lemma:Renes3} each $3$-torsion point $(x_i,y_i)\neq(0,0)$ gives a degree $3$ isogeny to an elliptic curve in triangular form $E'\colon y^2+\gamma_ixy+y=x^3$ with $\gamma_i=-3(2+ax_i)$. Thus, it is enough to prove that $\gamma_1,\gamma_2,\gamma_3$ are precisely the three roots of $\Gamma_1(a,Y)$.

First, we recall that the $x$-coordinates of the $3$-torsion points of $E$ are the roots of the $3$-division polynomial
\[
\psi_3(a,Z) = 3Z^4+a^2Z^3+3aZ^2+3Z.
\]
Notice, that $\psi_3(a,0)=0$, since $(0,0)$ is always a $3$-torsion point. Since we are not interested in the isogeny with kernel $\langle(0,0)\rangle$, we divide by $Z$ and consider the polynomial with integer coefficients
\[
\widetilde{\psi_3}(X,Z)=3Z^3+Z^2X^2+3XZ+3.
\]
The roots of $\widetilde{\psi_3}(a,Z)$ are the $x$-coordinates of the $3$-torsion points of $E$ distinct from $(0,0)$.
We also introduce the polynomial 
\[
\theta(X,Y,Z)=3XZ + Y + 6\in\mathbb{Z}[X,Y,Z]
\]
with the property that for any root $z$ of $\widetilde{\psi_3}(a,Z)$ the root of $\theta(a,Y,z)$ is one of $\gamma_1,\gamma_2,\gamma_3$ by Lemma~\ref{lemma:Renes3}.
Now, we consider the elimination ideal
$$I=(\widetilde{\psi_3}(X,Z),\theta(X,Y,Z)) \cap \mathbb{Q}[X,Y].$$
A computation shows that the ideal $I$ is principally generated by the polynomial 
$\Gamma_1(X,Y) = X^3Y^2 + 3X^3Y + 9X^3 - Y^3 - 18Y^2 - 108Y - 216$. Thus, the roots of $\Gamma_1(a,Y)$ are precisely $\gamma_1,\gamma_2,\gamma_3$ as required.
\end{proof}

  \begin{proposition}\label{prop:propertiesofGamma2}
        Let $E$ be a supersingular elliptic curve over $\mathbb{F}_{p^2}$ in triangular form.
           Let $\delta\in\mathbb{F}_{p^2}$ be a root of $\Gamma_2(X,j(E))$, then the elliptic curve $y^2+\delta xy+y=x^3$ has $j$-invariant $j(E)$.
       \end{proposition} 

\begin{proof}
Let  $y^2 + axy + y = x^3$  be the triangular form of $E$ with $a\in\mathbb{F}_{p^2}$, $a^3\neq27$. 
We recall that the $j$-invariant of $E$ is given by the formula 
$$j(E) = \frac{a^3(a^3 - 24)^3}{a^3 - 27}.$$
Thus, there are $12$ possible $a$'s giving an elliptic curve in triangular form with $j$-invariant $j(E)$. By clearing denominators, these can be expressed as the roots of the polynomial $\Gamma_2(X,j(E))$, where
\[
\Gamma_2(X,Y) = X^{12} - 72X^9 + 1728X^6 - X^3Y - 13824X^3 + 27Y\in \mathbb{Z}[X,Y].
\]
\end{proof}

\subsection{The modeling}
We assume the following setup.
We have two supersingular elliptic curves $E_{\start}$ and $E_{\finish}$ over $\mathbb{F}_{p^2}$ in triangular form that are connected via an isogeny $\varphi$ of degree $3^m$ ($m\in\mathbb{Z}_{\geq1}$) and suppose that the first $3$-isogeny in the decomposition of $\varphi$ has kernel $\neq\langle(0,0)\rangle$. 
We assume further that $|E_{\start}(\mathbb{F}_{p^2})| = (p - 1)^2$ if $p\equiv 1\mod 3$ and  $|E_{\start}(\mathbb{F}_{p^2})| = (p + 1)^2$ if $p\equiv 2\mod 3$.
Let $a_{\start}$ be the $xy$-coefficient of the triangular form of $E_{\start}$ and let $j(E_{\finish})$ be the $j$-invariant of $E_{\finish}$. 
Using Renes polynomials (Definition~\ref{def:Renesdegree3}), we construct a polynomial system whose solutions represent the $a$'s of the triangular forms of the elliptic curves in the path that connects $E_{\start}$ and $E_{\finish}$. Notice that Tate's theorem ensures that all curves in the path will have the same number of rational points as $E_{\start}$, and thus can be put in triangular form as well by Lemma~\ref{lem:3torsion-rationality}.

\begin{tcolorbox}[
  title = \hypertarget{system6}{Polynomial system 4: Renes Formulas for $3$-isogenies, complete},
  colframe=black,         
  colback=white,          
  coltitle=black,         
  colbacktitle=white,     
  boxrule=0.5pt,
  arc=4pt,
  fonttitle=\bfseries,
  sharp corners,
  top=1mm,
  bottom=0mm,
]
\raggedright $U_{m} = \mathbb{F}_{p^2}[a_1,\ldots,a_{m}]$\\
$\mathcal{F}_{3\text{-Renes},m} = [\Gamma_1(a_{\start},a_1),\Gamma_1(a_1,a_2),\cdots,\Gamma_1(a_{m-1},a_{m}),\Gamma_2(a_{m},j(E_{\finish}))] \subseteq U_{m}$

\begin{equation*}\label{eq_syst_Renes_3}
\begin{cases}
    &\Gamma_1(a_{\start},a_1) = 0\\  
    &\Gamma_1(a_1,a_2) = 0\\
    &\cdots \cdots\\
    &\Gamma_1(a_{m-2},a_{m-1}) = 0\\
    &\Gamma_1(a_{m-1},a_m) = 0\\
    &\Gamma_2(a_{m},j(E_{\finish})) = 0 
\end{cases}    
\end{equation*}
\end{tcolorbox}

The solutions of Polynomial system \hyperlink{system6}{4} are in bijection with the set $P(3,m,E_{\start},j(E_{\finish}))$ of Definition~\ref{def:P(d,m,E,j)set} which parametrizes all possible non-backtracking paths of degree $3$ isogenies of length $m$ that start from $E_{\start}$ and ends into a supersingular elliptic curve with the same $j$-invariant as $E_{\finish}$. We state this in the next theorem. The proof is analogous to the proof of Theorem~\ref{thm:fund_thm_deg2}, where Lemma~\ref{lem:3torsion-rationality},  Lemma~\ref{lemma:Renes3}, Proposition~\ref{prop:properties of Gamma1}, and Proposition~\ref{prop:propertiesofGamma2} shall be used.

\begin{theorem}
	Let $E_{\start},E_{\finish}$ be two supersingular elliptic curves in triangular form over $\mathbb{F}_{p^2}$ (with the $y$-coefficient equal to $1$). Assume that $|E_{\start}(\mathbb{F}_{p^2})| = (p - 1)^2$ if $p\equiv 1\mod 3$ and  $|E_{\start}(\mathbb{F}_{p^2})| = (p + 1)^2$ if $p\equiv 2\mod 3$. Then, given $m \in \mathbb{Z}_{\geq 1}$, there is a one to one correspondence between the set of $\mathbb{F}_{p^2}$-rational solutions to Polynomial system \hyperlink{system6}{4}
     and the set $P(3,m,E_{\start},j(E_{\finish}))$. 
\end{theorem}

As for the degree $2$ case, this complete polynomial system is difficult to solve with Gr\"obner bases because the polynomial $\Gamma_2(a_{m},j(E_{\finish}))$ has degree $12$. So, as before, we replace the complete system with 12 compact ones, replacing the last equation with $\Gamma_1(a_{m-1},a)$, where the $a$'s are the roots of the polynomial $\Gamma_2(X,j(E_{\finish}))$.
We obtain the following systems, which we will focus our analysis on.

\begin{tcolorbox}[
  title = \hypertarget{system7}{Polynomial system 5: Renes Formulas for $3$-isogenies, compact},
  colframe=black,         
  colback=white,          
  coltitle=black,         
  colbacktitle=white,     
  boxrule=0.5pt,
  arc=4pt,
  fonttitle=\bfseries,
  sharp corners,
  top=1mm,
  bottom=0mm,
]
\raggedright $U_{m-1} = \mathbb{F}_{p^2}[a_1,\ldots,a_{m-1}]$\\
$\mathcal{R}_{3,m} = [\Gamma_1(a_{\start},a_1),\Gamma_1(a_1,a_2),\cdots,\Gamma_1(a_{m-2},a_{m-1}),\Gamma_1(a_{m-1},a)] \subseteq U_{m-1}$

\begin{equation*}\label{eq_syst_Renes_3}
\begin{cases}
    &\Gamma_1(a_{\start},a_1) = 0\\  
    &\Gamma_1(a_1,a_2) = 0\\
    &\cdots \cdots\\
    &\Gamma_1(a_{m-2},a_{m-1}) = 0\\
    &\Gamma_1(a_{m-1},a) = 0 
\end{cases}    
\end{equation*}
\end{tcolorbox}

\section{Algebraic Analysis}\label{section:algebraicanalysis}
In this section, we study some algebraic properties of the systems introduced:
\begin{itemize}
    \item the system of modular polynomials $\mathcal{M}_{\ell,m}$,  where $\ell\neq p$ is a prime number (Polynomial system~\hyperlink{system1}{1});
    \item the system of (compact) Renes polynomials $\mathcal{R}_{2,m-1}$ in degree $2$ (Polynomial System~\hyperlink{system3}{3});
    \item the system of (compact) Renes polynomials $\mathcal{R}_{3,m}$ in degree $3$ (Polynomial System~\hyperlink{system7}{5}).
\end{itemize}
Notice, that for degree $2$ we consider the system $\mathcal{R}_{2,m-1}$ parametrizing paths of length $m-1$ since it has $m-1$ variables as the system  $\mathcal{R}_{3,m}$ of degree $3$ (which parametrizes paths of length $m$).

First, we prove that the systems are zero-dimensional, i.e., they admit a finite number of solutions over the algebraic closure $\overline{\mathbb{F}}_p$. For modular polynomials, this was already observed in \cite{Taketal}. First, we need the following lemma.

\begin{lemma}\label{lem:zero-dim}
Let $S = K[x_1,\ldots,x_n]$ be a polynomial ring over a field $K$ and let $I = (f_1,\ldots,f_{n+1}) \subseteq S$ be an ideal. If, up to reordering the generators of $I$, there exists a term order $<$ on $S$ such that $\LT_<(f_1),\ldots,\LT_<(f_n)$ are coprime and of degree $\geq 1$, then $I$ is zero-dimensional and $f_1,\ldots,f_n$ is a regular sequence of (maximal) length $n$ in $I$.
\end{lemma}
\begin{proof}
Since $\LT_<(f_1),\ldots,\LT_<(f_n)$ are coprime, the polynomials $f_1,\ldots,f_n$ form a reduced Gr\"obner basis for the ideal $\Tilde{I} = (f_1,\ldots,f_n)$. Moreover, since $\LT_<(f_1),\ldots,\LT_<(f_n)$ are precisely $n$ and of degree $\geq 1$, there exists a permutation $\sigma \in \mathfrak{S}_n$ such that $\LT_<(f_i) = x_{\sigma(i)}^{\deg(\LT_<(f_i))}$ for all $i = 1,\ldots,n$. So, the ideal $\Tilde{I}$ is zero-dimensional and $\LT_<(f_1),\ldots,\LT_<(f_n)$ is a regular sequence of maximal length $n$ in $\Tilde{I}$. Macaulay Basis Theorem (see e.g. \cite[Theorem~1.5.7]{Robbiano1}) yields that the ideal $\Tilde{I}$ is zero-dimensional, hence $I$ is zero-dimensional as well. Finally, by  \cite[Tutorial 54]{Robbiano2}, we obtain that $f_1,\ldots,f_n$ is a regular sequence of maximal length $n$ in $I$.
\end{proof}

\begin{proposition}\label{prop_zero_dim}
Let $m \geq 2$. The ideals $(\mathcal{M}_{\ell,m})$, $(\mathcal{R}_{2,m-1})$,  and $(\mathcal{R}_{3,m})$ are zero-dimensional and contain a regular sequence of (maximal) length $m-1$.
\end{proposition}
\begin{proof}
For every statement, we want to use Lemma \ref{lem:zero-dim} to conclude. For $(\mathcal{M}_{\ell,m})$, we consider a $\Lex$ term order $<$ with $j_{m-1} > j_{m-2} > \cdots > j_1$. Then, the leading terms of the first $m-1$ polynomials of $\mathcal{M}_{\ell,m}$ are $j_1^{\ell +1},j_2^{\ell + 1},\ldots,j_{m-1}^{\ell + 1}$, which are coprime and of degree $\geq1$. For $(\mathcal{R}_{2,m-1})$, we consider a $\Lex$ term order $<$ with $\alpha_{m-2} > \alpha_{m-3} > \cdots > \alpha_0$. Then, the leading terms of the first $m-1$ polynomials of $\mathcal{R}_{2,m-1}$ are $\alpha_0^{2},\alpha_1^{2},\ldots,\alpha_{m-2}^{2}$, which are coprime and of degree $\geq1$. For $(\mathcal{R}_{3,m})$ we consider a $\Lex$ term order $<$ with $a_{m-1} > a_{m-2} > \cdots > a_1$. Then, the leading terms of the first $m-1$ polynomials of $\mathcal{R}_{3,m}$ are $a_1^{3},a_2^{3},\ldots,a_{m-1}^{3}$, which are coprime and of degree $\geq1$.
\end{proof}

The notion of degree of regularity was introduced in \cite{BFS04} as a way to measure the complexity of solving a polynomial system via linear-algebra based algorithms. Indeed, under suitable assumptions the degree of regularity provides an upper bound on the solving degree \cite{Salizzoni}. In a nutshell, given a polynomial system $\mathcal{F}=\{f_1,\dots,f_s\}$ the degree of regularity $d_{\text{reg}}$ is the minimum integer $d$ such that the homogeneous degree $d$ piece of the ideal $(f_1^{\mathrm{top}},\dots,f_s^{\mathrm{top}})$ coincides with the polynomial ring in degree $d$. If this minimum does not exists, i.e., when the ideal $(f_1^{\mathrm{top}},\dots,f_s^{\mathrm{top}})$ is not zero-dimensional, then the degree of regularity is not defined or is set to be $\infty$, depending on the convention.
Thus, when $(f_1^{\mathrm{top}},\dots,f_s^{\mathrm{top}})$ is not zero-dimensional the degree of regularity does not provide an immediate insight on the complexity of solving the related system.
We prove that this is the case for the systems under study. More precisely, in Proposition~\ref{prop:Krull-dim} we compute the Krull dimension of the ideal generated by the $\mathrm{top}$ part of the system. We need a technical lemma.

\begin{lemma}\label{lem:Krull-dim}
    Let $K$ be a field, and let $S=K[x_1,\ldots,x_n]$ be a polynomial ring. We consider the monomial ideal
    $I=(x_1,x_1x_2,\ldots,x_{n-1}x_n,x_n)$. Then, we have
    \[\dim(S/I)=
    \begin{cases}
        \frac{n-2}{2} \quad \mathrm{ if } \ n \ \mathrm{ even }\\
        \frac{n-1}{2} \quad \mathrm{ if } \ n \ \mathrm{ odd }.
    \end{cases}
    \]
\end{lemma}
\begin{proof}
If $n = 1$ or $2$, then $\dim(S/I) = \dim(K) = 0$, since $K$ is a field. Let $n \geq 3$, then
$$\dim(S/I) = \dim(K[x_2,\ldots,x_{n-1}]/(x_2x_3,\ldots,x_{n-2}x_{n-1})).$$
Thus, for computing the Krull dimension we may replace the polynomial ring by $S=K[y_1,\ldots,y_t]$ and the ideal by $I= (y_1y_2,\ldots,y_{t-1}y_t)$.
Now, the dimension of $I$ can be computed as 
\[\dim S/I=t - \min\{\mathrm{ht}(P) \ | \ P \in \Min(I)\},
\]
where $\Min(I)$ is the set of minimal primes of $I$.
Now, let $t = 2k$ be even ($k \in \mathbb{Z}_{\geq 1}$). We claim that the ideal $\mathcal{P} = (y_2,y_4,\ldots,y_{2k})$ of height $k$ is the minimal prime ideal over $I$ that realizes the minimum above.

By \cite[Theorem 1.3.1, Corollary 1.3.4]{Herzog-Hibi} an irredundant presentation of $I$ is given by 
$$I = \bigcap\limits_{i=1,\ldots,s} V_i,$$
where each $V_i$ is a (prime) ideal of variables with $\text{ht}(V_i)=\mu(V_i)$.
Moreover, since $I$ is squarefree (hence radical), we can write
$$I = \bigcap\limits_{P \in \Min(I)} P.$$
For any $i=1,\ldots,2k-1$, we have $y_i \mid y_i y_{i+1}$ and $y_{i+1} \mid y_i y_{i+1}$ and either $i$ or $i+1$ is even, thus  the prime ideal $\mathcal{P}$ contains $I$. We claim that it does not exist a minimal prime ideal $\mathcal{P}'$ of $I$ such that $\text{ht}(\mathcal{P}') < \text{ht}(\mathcal{P})$. Once the claim is proved, then we are done since $\mathcal{P}$ will be minimal and of the smallest possible height.

Assume by contradiction, that such a $\mathcal{P}'$ exists. Thanks to the decompositions above, we have
$$I = \bigcap\limits_{i = 1,\ldots,s} V_i = \bigcap\limits_{P \in \Min(I)} P.$$
By properties of the irredundant decomposition (see \cite[Lemma 1.3.5]{Herzog-Hibi}), for every $P \in \Min(I)$ there exists $i \in \{1,\ldots,s\}$ such that $P = V_i$. Then, for some $h \in \{1,\ldots,s\}$, we have $\mathcal{P}' = V_h$.
In particular, $\mathcal{P}'$ is an ideal of variables, so $\mu(\mathcal{P}') = \text{ht}(\mathcal{P}') < k$.
This tells us that there exist two consecutive variables $y_c,y_{c+1}$ (for some $c \in \{1,\ldots,2k-1\}$) such that
$y_c,y_{c+1} \not\in \mathcal{P}'$. But $y_c y_{c+1} \in I$, so $I \not\subseteq \mathcal{P}'$, and this is a contradiction.
So, the Krull dimension of $\Tilde{I}$ is $t - k = 2k - k = k = \frac{n-2}{2}$. 

For $t = 2k - 1$ odd the argument is the same but $\mathcal{P} = (y_2,y_4,\ldots,y_{2k - 2})$ is the minimal prime ideal which realizes the height.
\end{proof}

\begin{proposition}\label{prop:Krull-dim}
Let $m \geq 2$. The (Krull) dimension of the ideals $(\mathcal{M}_{\ell,m}^\mathrm{top})$,$ (\mathcal{R}_{2,m-1}^{\mathrm{top}})$ and $(\mathcal{R}_{3,m}^{\mathrm{top}})$ coincides and it is equal to 
$$\begin{cases}
        \frac{m-2}{2} \quad \mathrm{ if } \ m \ \mathrm{ even }\\
        \frac{m-3}{2} \quad \mathrm{ if } \ m \ \mathrm{ odd }.
\end{cases}$$
In particular, if $m \geq 4$, these ideals are not zero-dimensional.
\end{proposition}
\begin{proof}
For every statement, we use Lemma~\ref{lem:Krull-dim} with $n=m-1$ to conclude. For $(\mathcal{M}_{\ell,m}^\mathrm{top})$, we have that $(\mathcal{M}_{\ell,m}^\mathrm{top}) = (j_1^{\ell+1},j_1^{\ell} j_2^{\ell},\ldots,j_{m-2}^{\ell} j_{m-1}^{\ell},j_{m-1}^{\ell+1})$. Passing to the radical the Krull dimension does not change and $\sqrt{(\mathcal{M}_{\ell,m}^\mathrm{top})} = (j_1,j_1j_2,\ldots,j_{m-2}j_{m-1},j_{m-1})$ has the shape of Lemma~\ref{lem:Krull-dim}. Similarly,  we have that  $\sqrt{(\mathcal{R}_{2,m-1}^{\mathrm{top}})} = (\alpha_0,\alpha_0\alpha_1,\ldots,\alpha_{m-3}\alpha_{m-2},\alpha_{m-2}),$ and
$\sqrt{(\mathcal{R}_{3,m}^{\mathrm{top}})}  =(a_1,a_1a_2,\ldots,\allowbreak a_{m-2}a_{m-1},a_{m-1})$. 
\end{proof}

Another consequence of Proposition~\ref{prop:Krull-dim} is that when $m\geq4$ the polynomial systems are not semi-regular \cite{BFS04,BDDMMT21,Pardue}. 

\begin{corollary}
    For every $m\geq4$, we have that $\mathcal{M}_{\ell,m}^\mathrm{top}$, $ \mathcal{R}_{2,m-1}^{\mathrm{top}}$ and $\mathcal{R}_{3,m}^{\mathrm{top}}$ are not semi-regular sequences.
\end{corollary}

Another invariant which can be used as a proxy for the solving degree is the Castelnuovo--Mumford regularity. Indeed, when the ideal is in generic coordinates in the sense of Bayer and Stillman \cite{BS87}, then the regularity of the ideal generated by the homogenized polynomials of the system is greater or equal than the solving degree of the system by \cite{Caminata-Gorla2}.
On the other hand, in Proposition~\ref{prop:genericcoordinates}, we prove that the systems $ \mathcal{M}_{\ell,m}$, $\mathcal{R}_{2,m-1}$ and $\mathcal{R}_{3,m}$ are not in generic coordinates.

For the convenience of the reader, we recall here the relevant definitions and notations.
Let $K$ be a field and let $S=K[x_1,\dots,x_n]$ be a polynomial ring. For a polynomial $f\in S$, we denote by $f^h\in R=S[t]$ its homogenization with respect to an extra variable $t$, that is
$$f^h = t^{\deg(f)} \cdot f\left(\frac{x_1}{t},\ldots,\frac{x_n}{t}\right).$$
For a list of polynomials $\mathcal{F} = [f_1,\ldots,f_s]$ in $S$, we denote by 
$\mathcal{F}^h = [f_1^h,\ldots,f_s^h]$ 
 the list of the homogenized polynomials in $R$.
If $I \subseteq S$ is an ideal, we denote by $I^h$ the ideal generated by all the homogenizations of elements in $I$, that is $I^h = (f^h \ | \ f \in I)$. Clearly, we have $(\mathcal{F}^h)\subseteq (\mathcal{F})^h$, but the inclusion may be strict.

Given a homogeneous ideal $J$ of $R$, the \emph{saturation} of $J$ with respect to the irrelevant maximal ideal $\mathfrak{m}=(x_1,\ldots,x_n,t)$ of $R$ is
$$J^{\sat} = \bigcup\limits_{d \geq 0} \{f \in R \ | \ fm \in J \  \forall m \in R_d\}=\bigcup\limits_{d \geq 0}(J:\mathfrak{m}^d).$$
We stress that the condition  $fm \in J \  \forall m \in R_d$  can be checked only on the monomials of $R_d$, since they form a $K$-basis for $R_d$ and $J$ is an ideal. Moreover, we always have that $J \subseteq J^{\sat}$.
Now, let $\overline{J} = J \overline{K}[x_1,\ldots,x_n,t]$ be the extension of the ideal $J$ over the algebraic closure $\overline{K}$ of $K$. We say that $J$ is in \emph{generic coordinates} over $\overline{K}$ if $|\mathcal{Z}_+(\overline{J})| < \infty$ and either $|\mathcal{Z}_+(\overline{J})| = 0$ or $t \nmid 0$ modulo $\overline{J}^{\sat}$.

\begin{lemma}\label{lem:gen-coord}
Let $K$ be an algebraically closed field and let $S = K[x_1,\ldots,x_n]$ be a polynomial ring. Let $\mathcal{F} = [f_1,\ldots,f_m]$ be a list of non-homogeneous polynomials such that $|\mathcal{Z}(\mathcal{F})| = 1$. 
Assume that there exists $i \in \{1,\ldots,n\}$ such that for all integers $a\geq1$  we have $x_i^a \not\in \Supp(\mathcal{F}^h)  = \bigcup\limits_{c = 1}^m \Supp(f_c^h)$. Then, the ideal $(\mathcal{F}^h)$ is not in generic coordinates over $K$.
\end{lemma}
\begin{proof}
Let $\mathcal{Z}(\mathcal{F})=\{(a_1,\ldots,a_n)\}\subseteq
K^n$. Thanks to the Shape Lemma (see e.g. \cite{Caminata-Gorla2}), we get that $\LexGB((\mathcal{F})) = [x_1 - a_1,\ldots,x_n-a_n]$. Since $(\mathcal{F})$ is zero-dimensional, $[x_1 - a_1,\ldots,x_n-a_n]$ is also the reduced Gr\"obner basis with respect to $\DegRevLex$. Thanks to \cite[Proposition 4.3.21]{Robbiano2}, the homogenized ideal $(\mathcal{F})^h$ is generated by the homogenization of this basis, that is
$$(\mathcal{F})^h = (x_1 - a_1t,\ldots,x_n - a_nt) \subseteq S[t] = R.$$
From now on, we equip $R$  with the $\DegRevLex$ term order with $x_1 > \cdots > x_n > t$. Moreover, we have that $(\mathcal{F}^h) \subseteq (\mathcal{F}^h)^{\mathrm{sat}} \subseteq (\mathcal{F})^h$, where the former follows from the definition of saturation and the latter from the fact that
$$(\mathcal{F}^h)^{\mathrm{sat}} = \bigcup\limits_{d \geq 0}((\mathcal{F}^h):\mathfrak{m}^d), \ \ \ \ \ \ \ \  (\mathcal{F})^h = \bigcup\limits_{d \geq 0}((\mathcal{F}^h) : (t)^d), \ \ \ \ \ \ \ \ \  \mbox{and}\ (t) \subseteq \mathfrak{m}.$$
Now, we claim that $t \mid 0$ modulo $(\mathcal{F}^h)^{\mathrm{sat}}$. By assumption, we can fix $i \in \{1,\ldots,n\}$ such that $x_i^a \not\in \Supp(\mathcal{F}^h) \ \forall a \in \mathbb{Z}_{\geq 1}$ and we know that $x_i - a_it \in (\mathcal{F})^h$. So, again since $(\mathcal{F})^h = \bigcup\limits_{d \geq 0}((\mathcal{F}^h):(t)^d)$,  there exists $\Tilde{d} \in \mathbb{Z}_{\geq 1}$ such that
$$(x_i - a_it)\cdot t^{\Tilde{d}} \in (\mathcal{F}^h) \subseteq (\mathcal{F}^h)^{\mathrm{sat}}.$$
Note that $\Tilde{d}$ cannot be $0$. In fact, since by assumption $x_i^a \not\in \Supp(\mathcal{F}^h)$ for all $a \in \mathbb{Z}_{\geq 1}$, then $x_i \in \Supp(x_i - a_it)$ cannot be written as a polynomial combination of the generators $\mathcal{F}^h$.

If we prove that $x_i - a_it \not\in (\mathcal{F}^h)^{\mathrm{sat}}$, the claim follows. Suppose by contradiction that $x_i - a_it \in (\mathcal{F}^h)^{\mathrm{sat}}$. Then, by definition of saturation, there exists $d' \in \mathbb{Z}_{\geq0}$ such that $x_i - a_it \in [(\mathcal{F}^h) : \mathfrak{m}^{d'}]$. Now, $x_i^{d'} \in \mathfrak{m}^{d'},$ so
$$f = (x_i - a_it) \cdot x_i^{d'} = x_i^{d' + 1} - a_i x_i^{d'}t \in (\mathcal{F}^h).$$
But, since by assumption  $x_i^a \not\in \Supp(\mathcal{F}^h)$ for all $a \in \mathbb{Z}_{\geq_1}$, again $x_i^{d' + 1} \in \Supp(f)$ cannot be written as a polynomial combination of the generators $\mathcal{F}^h$. Therefore $f \not\in (\mathcal{F}^h)$ and we get a contradiction, which proves the claim.

Finally, since $(a_1,\ldots,a_n) \in \mathcal{Z}(\mathcal{F}),$ we get that $[a_1 : \cdots : a_n : 1] \in \mathcal{Z}_+(\mathcal{F}^h),$ so $|\mathcal{Z}_+(\mathcal{F}^h)| > 0$. This concludes the proof.
\end{proof}

\begin{proposition}\label{prop:genericcoordinates}
Let $m \geq 4$ and let $\mathcal{F} = \mathcal{M}_{\ell,m}$, $\mathcal{R}_{2,m-1}$, or $\mathcal{R}_{3,m}$. If $|\mathcal{Z}(\mathcal{F})| = 1$, then $(\mathcal{F}^h)$ is not in generic coordinates over $\overline{\mathbb{F}}_p$.
\end{proposition}
\begin{proof}
Let $\mathcal{F} = \mathcal{M}_{\ell,m}$. Since $m \geq 4$, we have that $j_2^b \not\in \Supp(\mathcal{M}_{\ell,m}^h)$ for all $b \in \mathbb{Z}_{\geq 1}$. So, we conclude using Lemma \ref{lem:gen-coord}. Similarly, for $\mathcal{R}_{2,m-1}$, it holds $\alpha_1^b \not\in \Supp(\mathcal{M}_{2,m-1}^{h})$ for all $b \in \mathbb{Z}_{\geq 1}$ and for $\mathcal{R}_{3,m}$ it holds $a_2^b \not\in \Supp(\mathcal{M}_{3,m}^{h})$ for all $b \in \mathbb{Z}_{\geq 1}$, so in both cases we conclude using again Lemma \ref{lem:gen-coord}.
\end{proof}

\begin{remark}
If $m \geq 6$ for $\mathcal{F} = \mathcal{M}_{\ell,m}$, $\mathcal{R}_{2,m-1}$, or $\mathcal{R}_{3,m}$, it is easier to check that $(\mathcal{F}^h)$ is not in generic coordinates over $\overline{\mathbb{F}}_p$. In fact, one can show that $|\mathcal{Z}_+(\overline{\mathcal{F}^h})| = \infty$. Indeed, $[0 : v_1 : 0 : \cdots : 0 : v_2 : 0] \in \mathcal{Z}_+(\overline{\mathcal{F}^h})$ for all $v_1$,$v_2\in\overline{\mathbb{F}}_p$ linearly independent. 
\end{remark}

\section{Experimental Results}\label{section:experiments}

We solved the polynomial systems arising from the previous modelings using Gr\"obner bases in order to compare the modeling based on Renes polynomials with that based on modular polynomials. All computations were performed using Magma V2.28-13 \cite{Magma}, on a Dell Inc. Precision 7875 Tower, with 64 GB of RAM and processor AMD \textregistered  Ryzen threadripper pro 7975wx 32-cores x 64.

We constructed the test examples as follows. Let $d=2,3$ be the degree. We selected a random supersingular elliptic curve $E_{\start}$ in Montgomery form (for degree $d=2$) or triangular form (for degree $d=3$) and we constructed a non-backtracking path of length $m$ of isogenies obtaining the curve $E_{\finish}$. Then, we built the system $\mathcal{M}_{d,m}$ of modular polynomials (Polynomial system~\hyperlink{system1}{1}) and $\mathcal{R}_{d,m}$ of Renes polynomials (Polynomial System~\hyperlink{system3}{3} and Polynomial System~\hyperlink{system7}{5}) to recover the path from $E_{\start}$ to $E_{\finish}$. These systems were solved using Magma's implementation of F4 \cite{FaugereF4}. We recorded the CPU time and the maximal degree reached during the Gr\"obner basis computation (used as a proxy for the solving degree).

We repeated these computations for nine different values of the underlying prime $p$, ranging from $10$ to $32$ bits. For each set of parameters, we performed 10 tests. In Table~\ref{table:deg2} and Table~\ref{table:deg3}, we report the results for degrees $2$ and $3$, respectively. The results show that the solving degree is consistently smaller for Renes polynomials than for modular polynomials. Moreover, the average solving times are two to three orders of magnitude lower for Renes polynomials. This allows us to solve paths of isogenies of length up to $16$ in degree $2$ for small primes, whereas we are limited to length $12$ when using modular polynomials. A similar trend is observed in degree $3$.

\begin{remark}
To solve longer paths of degree 3 isogenies described by the modular polynomial systems $\mathcal{M}_{3,m}$, Takahashi et al. \cite{Taketal} split the system into two halves and compute a Gröbner basis for each of them. Then, by eliminating variables, they find the solutions. This technique can also be applied to the polynomial systems $\mathcal{R}_{2,m}$ and $\mathcal{R}_{3,m}$ with Renes polynomials, allowing one to recover longer isogeny paths. However, to enable a clearer comparison of the systems, we chose not to apply this splitting technique to any of the computations performed in this paper.
\end{remark}

\newpage
\small
\begin{longtable}{|c|c|c|c|c|c|c|c|c|c|}
\toprule
$m$ & $\mathrm{tm}(\mathcal{M}_{2,m})$ & $\mathrm{tm}(\mathcal{R}_{2,m})$  & $\SolvingDegree(\mathcal{M}_{2,m})$& $\SolvingDegree(\mathcal{R}_{2,m})$ &
$m$ & $\mathrm{tm}(\mathcal{M}_{2,m})$ & $\mathrm{tm}(\mathcal{R}_{2,m})$  & $\SolvingDegree(\mathcal{M}_{2,m})$& $\SolvingDegree(\mathcal{R}_{2,m})$ \\ \hline
\multicolumn{5}{|c|}{$p = 587417, 20 \mbox{ bits}$} & \multicolumn{5}{|c|}{$p = 919447, 20 \mbox{ bits}$} \\
\midrule
\endfirsthead

6 & 0.02 & 0.003 & 6 & 4 & 6 & 0.02 & 0.002 & 6 & 4 \\
7 & 0.514 & 0.018 & 7 & 5 & 7 & 0.671 & 0.02 & 7 & 5 \\
8 & 1.722 & 0.09 & 7 & 5 & 8 & 2.612 & 0.102 & 7 & 5 \\
9 & 66.563 & 0.298 & 8 & 5 & 9 & 135.981 & 0.375 & 8 & 5 \\
10 & 78.454 & 2.008 & 8 & 5 & 10 & 179.822 & 3.009 & 8 & 5 \\
11 & 3040.88 & 15.958 & 9 & 5 & 11 & 10055.526 & 27.815 & 9 & 5 \\
12 & 3359.307 & 50.588 & 9 & 5 & 12 & 11223.916 & 96.897 & 9 & 5 \\
13 & $>7$ hrs & 104.499 &  & 6 & 13 & $>7$ hrs & 213.71 &  & 6 \\
14 &  & 434.435 &  & 6 & 14 &  & 1016.976 &  & 6 \\
15 &  & 7632.326 &  & 6 & 15 &  & $>7$ hrs &  &  \\
16 &  & 23416.61 &  & 6 & 16& & & & \\
17 &  & $>7$ hrs &  &   & 17 & & & & \\
\midrule

\multicolumn{5}{|c|}{$p = 15739441, 24 \mbox{ bits}$} & \multicolumn{5}{|c|}{$p = 16541149, 24 \mbox{ bits}$} \\
\midrule

6 & 0.02 & 0.001 & 6 & 4 & 6 & 0.02 & 0.001 & 6 & 4 \\
7 & 0.502 & 0.018 & 7 & 5 & 7 & 0.506 & 0.019 & 7 & 5 \\
8 & 1.738 & 0.09 & 7 & 5 & 8 & 1.716 & 0.08 & 7 & 5 \\
9 & 67.766 & 0.309 & 8 & 5 & 9 & 67.375 & 0.276 & 8 & 5 \\
10 & 76.685 & 2.019 & 8 & 5 & 10 & 77.189 & 1.951 & 8 & 5 \\
11 & 2710.243 & 15.68 & 9 & 5 & 11 & 2692.042 & 15.635 & 9 & 5 \\
12 & 2618.765 & 50.162 & 9 & 5 & 12 & 2610.953 & 50.164 & 9 & 5 \\
13 & $>7$ hrs & 104.273 &  & 6 & 13 & $>7$ hrs & 103.944 &  & 6 \\
14 &  & 437.867 &  & 6 & 14 &  & 438.755 &  & 6 \\
15 &  & 7582.87 &  & 6 & 15 &  & 7703.542 &  & 6 \\
16 &  & 23514.02 &  & 6 & 16 &  & $>7$ hrs &  &  \\
17 & & $>7$ hrs &   &   & 17 &  &          &  & \\

\midrule

\multicolumn{5}{|c|}{$p = 266545607, 28 \mbox{ bits}$} & \multicolumn{5}{|c|}{$p = 2757059413, 32 \mbox{ bits}$} \\
\midrule

6 & 0.022 & 0.003 & 6 & 4 & 6 & 0.086 & 0.01 & 6 & 4 \\
7 & 0.704 & 0.02 & 7 & 5 & 7 & 2.333 & 0.06 & 7 & 5 \\
8 & 2.674 & 0.102 & 7 & 5 & 8 & 8.678 & 0.387 & 7 & 5 \\
9 & 129.268 & 0.371 & 8 & 5 & 9 & 390.104 & 1.361 & 8 & 5 \\
10 & 168.687 & 2.923 & 8 & 5 & 10 & 476.839 & 10.147 & 8 & 5 \\
11 & 9739.063 & 27.347 & 9 & 5 & 11 & 21743.386 & 86.571 & 9 & 5 \\
12 & 11001.887 & 95.651 & 9 & 5 & 12 & 22894.691 & 284.438 & 9 & 5 \\
13 & $>7$ hrs & 207.348 &  & 6 & 13 & $>7$ hrs & 612.507 &  & 6 \\
14 &  & 1000.487 &  & 6 & 14 &  & 2762.891 &  & 6 \\
15 &  & $>7$ hrs &  &  & 15 &  & $>7$ hrs &  &  \\
\hline

\caption{{\small Timings and solving degrees for polynomial systems modeling a SIP of degree $2^m$. The values $\mathrm{tm}(\mathcal{M}_{2,m})$ and $\mathrm{tm}( \mathcal{R}_{2,m})$ record the average time (in seconds) of solving the polynomial systems $\mathcal{M}_{2,m}$ (modular polynomials) and $\mathcal{R}_{2,m}$ (Renes polynomials) modeling a SIP of degree $2^m$; $\SolvingDegree(\mathcal{M}_{2,m})$ and $\SolvingDegree(\mathcal{R}_{2,m})$  are the highest step degrees obtained during the computation by using Magma F4 implementation.}}\label{table:deg2}
\end{longtable}
\normalsize

\newpage
\small
\begin{longtable}{|c|c|c|c|c|c|c|c|c|c|}
\toprule
$m$ & $\mathrm{tm}(\mathcal{M}_{3,m})$ & $\mathrm{tm}(\mathcal{R}_{3,m})$  & $\SolvingDegree(\mathcal{M}_{3,m})$& $\SolvingDegree(\mathcal{R}_{3,m})$ &
$m$ & $\mathrm{tm}(\mathcal{M}_{3,m})$ & $\mathrm{tm}(\mathcal{R}_{3,m})$  & $\SolvingDegree(\mathcal{M}_{3,m})$& $\SolvingDegree(\mathcal{R}_{3,m})$ \\ \hline

\multicolumn{5}{|c|}{$p = 919447, 20 \mbox{ bits}$} & 
\multicolumn{5}{|c|}{$p = 15739441, 24 \mbox{ bits}$} \\
\midrule
6 & 0.901 & 0.074 & 8 & 6 & 6 & 0.648 & 0.06 & 8 & 6 \\
7 & 76.149 & 1.345 & 10 & 7 & 7 & 43.406 & 0.916 & 10 & 7 \\
8 & 1887.466 & 30.398 & 10 & 7 & 8 & 724.222 & 17.511 & 10 & 7 \\
9 & $>7$ hrs & 42.104 &  & 8 & 9 & 11580.257 & 203.132 & 11 & 8 \\
10 &  & 2620.534 &  & 8 & 10 & $m>9$ & 1000.555 &  & 8 \\
11 &  & $>7$ hrs &  &  & 11 &  & 17300.334 &  & 9 \\
12 &  &  &  &  & 12 &  & $>7$ hrs &  &  \\

\midrule
\multicolumn{5}{|c|}{$p = 16541149, 24 \mbox{ bits}$} & 
\multicolumn{5}{|c|}{$p = 12239911, 24 \mbox{ bits}$} \\
\midrule
6 & 0.652 & 0.06 & 8 & 6 & 6 & 0.906 & 0.074 & 8 & 6 \\
7 & 43.066 & 0.912 & 10 & 7 & 7 & 74.625 & 1.304 & 10 & 7 \\
8 & 726.928 & 17.443 & 10 & 7 & 8 & 1721.217 & 29.82 & 10 & 7 \\
9 & 11609.202 & 202.294 & 11 & 8 & 9 & $>7$ hrs & 411.02 &  & 8 \\
10 & $>7$ hrs & 1001.411 &  & 8 & 10 &  & 2531.694 &  & 8 \\
11 &  & 17303.164 &  & 9 & 11 &  & $>7$ hrs &  &  \\
12 &  & $>7$ hrs &  &  & 12 &  &  &  &  \\

\midrule
\multicolumn{5}{|c|}{$p = 266545607, 28 \mbox{ bits}$} & 
\multicolumn{5}{|c|}{$p = 2757059413, 32 \mbox{ bits}$} \\
\midrule
6 & 0.897 & 0.073 & 8 & 6 & 6 & 2.968 & 0.283 & 8 & 6 \\
7 & 75.404 & 1.302 & 10 & 7 & 7 & 241.397 & 4.486 & 10 & 7 \\
8 & 1766.141 & 30.031 & 10 & 7 & 8 & 4705.333 & 95.73 & 10 & 7 \\
9 & $>7$ hrs & 419.066 &  & 8 & 9 & $>7$ hrs & 1225.204 &  & 8 \\
10 &  & 2579.288 &  & 8 & 10 &  & 6568.001 &  & 8 \\
11 &  & $>7$ hrs &  &  & 11 &  & $>7$ hrs &  &  \\

\bottomrule

\caption{{\small Timings and solving degrees for polynomial systems modeling a SIP of degree $3^m$. The values $\mathrm{tm}(\mathcal{M}_{3,m})$ and $\mathrm{tm}(\mathcal{R}_{3,m})$ record the average time (in seconds) of solving the polynomial systems $\mathcal{M}_{3,m}$ (modular polynomials) and $\mathcal{R}_{3,m}$ (Renes polynomials) modeling a SIP of degree $3^m$. $\SolvingDegree(\mathcal{M}_{3,m})$ and $\SolvingDegree(\mathcal{R}_{3,m})$ are the highest step degrees obtained during the computation by using Magma F4 implementation.}}\label{table:deg3}
\end{longtable}
\normalsize

\bibliographystyle{siam}
\bibliography{biblio.bib}

\end{document}